\documentclass[11pt,a4paper]{article}
\pdfoutput=1 
\usepackage{fullpage}

\usepackage{amsmath}
\usepackage{amsfonts}
\usepackage{amssymb}
\usepackage{amsthm}
\usepackage{thmtools}
\allowdisplaybreaks
\usepackage{mathtools}
\usepackage{setspace}
\usepackage{float}

\usepackage[T1]{fontenc}
\usepackage[tt=false]{libertine}
\usepackage[libertine]{newtxmath}
\usepackage{hyperref}
\usepackage[svgnames]{xcolor}
\hypersetup{colorlinks={true},urlcolor={blue},linkcolor={DarkBlue},citecolor=[named]{DarkGreen}}
\usepackage[authoryear,square]{natbib}
\usepackage[capitalise,nameinlink,noabbrev]{cleveref}
\usepackage{xspace}

\usepackage{doi}
\usepackage{bm}
\usepackage{bbm}

\usepackage{booktabs} 
\usepackage{authblk} 
\usepackage{graphicx} 

\usepackage{tcolorbox}

\newtheorem{theorem}{Theorem}[section]

\newtheorem{lemma}[theorem]{Lemma}

\newtheorem{inftheorem}{Informal Theorem}
\newtheorem{infcorollary}{Informal Corollary}

\theoremstyle{definition}
\newtheorem*{comment*}{Comment}
\newtheorem{definition}[theorem]{Definition}
\newtheorem{remark}[theorem]{Remark}

\newcommand{\pX}{\mathbf{X}}

\usepackage{xcolor}
\usepackage{xspace}

\newcommand{\efkx}{{EFkX}\xspace}
\newcommand{\efx}{{EFX}\xspace}
\newcommand{\aefkx}{{$\frac{k+1}{k+2}$-EFkX}\xspace}

\usepackage{booktabs} 
\usepackage[ruled,vlined]{algorithm2e} 

\SetAlFnt{\small}
\SetAlCapFnt{\small}
\SetAlCapNameFnt{\small}
\SetAlCapHSkip{0pt}
\IncMargin{-\parindent}

\SetCommentSty{mycommfont}

\SetKwInput{KwData}{Input}
\SetKwInput{KwResult}{Output}

\crefname{claim}{Claim}{Claims}

\usepackage{enumitem}

\newlist{threecases}{enumerate}{1}
\setlist[threecases]{label={(\alph*)}}
\crefname{threecasesi}{Case}{Cases}

\newlist{properties}{enumerate}{1}
\setlist[properties]{label={(\alph*)}}
\crefname{propertiesi}{Property}{Properties}

\newlist{propertiesss}{enumerate}{1}
\setlist[propertiesss]{label={(\roman*)}}
\crefname{propertiesssi}{Property}{Properties}

\newlist{pproperties}{enumerate}{1}
\setlist[pproperties]{label={(F.\arabic*)}}
\crefname{ppropertiesi}{Property}{Properties}

\newlist{sproperties}{enumerate}{1}
\setlist[sproperties]{label={(S.\arabic*)}}
\crefname{spropertiesi}{Step Property}{Step Properties}

\newlist{levels}{enumerate}{1}
\setlist[levels]{label={(L.\arabic*)}}
\crefname{levelsi}{Level}{Levels}

\usepackage{tcolorbox}

\usepackage{bm}

\crefname{appendixsection}{Appendix}{Appendices}

\setcitestyle{authoryear}

\xspaceaddexceptions{]\}}

\theoremstyle{remark}

\def \X {\mathbf{X}\xspace}

\newcommand{\Gr}{G_\mathrm{mod}\xspace}
\newcommand{\Ge}{G_\mathrm{mod}\xspace}

\newcommand{\PP}{\mathcal{P}\xspace}

\newcommand{\CR}{\textnormal{\textsc{CycleResolution}}\xspace}
\newcommand{\ACR}{\textnormal{\textsc{AllCyclesResolution}}\xspace}
\newcommand{\PR}{\textnormal{\textsc{PathResolution}}\xspace}
\newcommand{\PRPA}{\textnormal{\textsc{PathResolution$^*$}}\xspace}
\newcommand{\PPA}{\textnormal{\textsc{3PA}}\xspace}

\newcommand{\UCG}{\textnormal{\textsc{UncontestedCritical}}\xspace}

\newcommand{\IFAA}{\textnormal{\textsc{ImprovedFewAgentsAllocate}}\xspace}

\newcommand{\AEC}{\textnormal{\textsc{AllocateAndEliminateCritical}}\xspace}
\newcommand{\AEFKX}{\textnormal{\textsc{ApproximateEFkX}}\xspace}

\newcommand{\GPPA}{\textnormal{\textsc{G3PA}}\xspace}

\newcommand{\GPPAplus}{\textnormal{\textsc{G3PA+}}\xspace}

\newcommand{\ACC}{\textnormal{\textsc{ContestedCritical}}\xspace}

\newcommand{\UCC}{\textnormal{\textsc{UncontestedCritical}}\xspace}

\newcommand{\PPAplus}{\textnormal{\textsc{3PA$^{+}$}}\xspace}

\newcommand{\ECE}{\textnormal{\textsc{EnvyCycleElimination}}\xspace}

\renewcommand{\paragraph}[1]{\medskip\noindent\textbf{#1.\;}}

\newenvironment{ack}{\section*{Acknowledgments}}{}

\author[1]{Aris Filos-Ratsikas}
\author[2]{Georgios Kalantzis}
\author[3]{Fangxiao Wang}

\affil[1]{Department of Computing\\
  Imperial College London\\
  \texttt{Aris.Filos-Ratsikas@imperial.ac.uk}}
\affil[2]{University of Edinburgh\\
  \texttt{G.Kalantzis@sms.ed.ac.uk}}
\affil[3]{Department of Computing\\
  The Hong Kong Polytechnic University\\
  \texttt{fangxiao.wang@connect.polyu.hk}}

\date{}

\title{Approximate Envy-Free Allocations up to any $k$ Goods}

\begin{document}

\maketitle

\begin{abstract}
We study the problem of finding approximate envy-free allocations up to any $k$ goods ($\alpha$-\efkx), when agents have additive values over goods in a bundle. As our main result, we show that for any fixed integer $k\geq2$, $\frac{k+1}{k+2}$-EFkX allocations exist for any number of agents, and can be computed in polynomial time, via an appropriate generalization of the \PPA algorithm of \citep{amanatidis2024pushing}. An immediate corollary of this result is that $3/4$-EF2X allocations exist for any number of agents; in contrast, $2/3$-EFX allocations are only known to exist for up to 7 agents. We improve this latter result by devising an algorithm that achieves $2/3$-EFX for 8 agents. We also consider \efkx graph orientations; we prove that such orientations do not always exist, and that deciding their existence is NP-complete, thereby generalizing the corresponding result of \citep{christodoulou2023fair} for $k=1$.
\end{abstract}

\section{Introduction}

The area of computational fair division studies the problem of allocating resources to agents in a fair manner, based on those agents' heterogeneous and often contrasting preferences over how the resources should be allocated. Amongst its many variants, the one that is currently at the forefront of the associated literature is the existence of fair allocations of \emph{indivisible goods}, also known as \emph{discrete fair division} \citep{amanatidis2023fair}. This literature has studied multiple notions of fairness for this setting, but arguably the most intriguing out of all those is that of \emph{envy-freeness up to any good (EFX)} \citep{caragiannis2019unreasonable}; this notion stipulates that the envy of one agent for another agent's allocated bundle is eliminated once we remove the least valuable good, from the envious agent's perspective, from the bundle of the envied agent.

The question of whether EFX allocations always exist, and, if they do, whether they can be computed in polynomial time, is still open. In fact, it would be fair to say that this is the most important question, or the ``holy grail'', of this very active research area. In the absence of a definitive answer, the literature has produced a plethora of results on meaningful special cases, which can be seen as important building blocks towards tackling the main question. For example, in a breakthrough result, \citet{chaudhury2024efx} showed that EFX allocations exist for three agents (strengthening a previous existence result for two agents due to \cite{PR18}, see also \citep{akrami2025efx}). \citet{amanatidis2021maximum} showed that EFX allocations exist and can be computed in polynomial time on instances where agents can have one of two possible values for the goods, and \citet{christodoulou2023fair} showed a corresponding result when the values can be represented by a graph. \citet{PR18} showed an EFX existence result when the agents' values induce the same ordering over goods. 

Achieving \emph{exact} EFX allocations has proven to be quite challenging for settings that go beyond the aforementioned special cases. For this reason, the literature has applied further relaxations to the EFX notion, two of the most prominent of which are the following:
\begin{itemize}
    \item[-] \emph{approximate-EFX or $\alpha$-EFX}, i.e., a notion that stipulates that after the removal of any good from the envied agent's bundle, the remaining envy is bounded multiplicatively by a factor $\alpha$. 
    \item[-] \emph{\efkx}, i.e., a notion that stipulates that the envy is eliminated after the removal of any $k$ goods from the envied agent's bundle.
\end{itemize}
Obviously, when $\alpha=1$ and $k=1$, we recover the definition of EFX allocations. The notion of $\alpha$-EFX was notably studied by \citet{ANM2019}, who proved the (currently best known) approximation of $\alpha=0.618$; this bound was established also by \citet{farhadi2021almost} with a different algorithm. \citet{markakis2023improv} improved the factor to $\alpha=2/3$ when all of the $n$ agents agree on the $n$ most valuable goods. In a work that is most relevant to us, \citet{amanatidis2024pushing} obtained $2/3$-EFX allocations for three important cases, namely when (i) there are at most $7$ agents, or (ii) when the values can be represented by a multigraph, or (iii) when the agents have one of three possible values.

The notion of \efkx allocations was defined by \citet{akramiEF2xallocation}, who, among other results, proved the existence of EF2X allocations for a rather restricted class of values (namely when each agent has a value in $\{0,v_g\}$ for some good $g$). Recently, \citet{ashuri2025ef2x} showed that EF2X allocations exist for up to $4$ agents. To the best of our knowledge, prior to our work, results for $k>2$ were not known in the literature.  

\subsection{Our Contribution}
In this work, for the first time we combine the two aforementioned lines of investigation, and consider \emph{approximate-\efkx allocations}. This combination allows us to achieve significant improvements to the approximation factor $\alpha$ compared to the state-of-the-art results for EFX, without any restrictions on the number of agents or their values. Our main result is given by the following theorem, informally stated below.

\begin{inftheorem}\label{infthm:main-theorem}
For every fixed integer $k\geq 2$ and any number of agents with
additive valuations, a \aefkx allocation exists and can be computed
in polynomial time.
\end{inftheorem}

\noindent To put \cref{infthm:main-theorem} in context, we compare its instantiation for $k=2$ with the aforementioned results of \cite{akramiEF2xallocation} and \cite{ashuri2025ef2x}, which establish the existence of EF2X allocations for rather restricted valuation classes, and only for $4$ agents, respectively. By allowing an approximation of $\alpha=3/4$, we are able to achieve $(3/4)$-EF2X allocations for \emph{any} number of agents and \emph{any} additive valuation functions. We state the corresponding corollary below.

\begin{infcorollary}\label{infcor:ef2x}
For any number of agents with additive values, a $3/4$-EF2X allocation exists and can be computed in polynomial time.
\end{infcorollary}

\noindent \cref{infthm:main-theorem} does not apply when $k=1$. For this case, we manage to improve on the result of \citet{amanatidis2024pushing} by showing the existence of $2/3$-EFX allocations for $8$ agents. Although the improvement from $7$ to $8$ agents may not seem that significant at first glance, the extension reveals some of the inherent difficulties of moving beyond the case of $7$ agents, and introduces new techniques to achieve that. We state the corresponding theorem informally below. 

\begin{inftheorem}\label{infthm:8-agents}
A $2/3$-EFX allocation exists and can be computed in polynomial time for instances with up to $8$ agents. 
\end{inftheorem}

\noindent We also consider \efkx allocations on graphs, extending the setting introduced by \citet{christodoulou2023fair} in the context of EFX allocations. In this setting, every good (corresponding to an edge of the graph) is positively valued by only two agents (the endpoints of the edge in the graph). We are particularly interested in \emph{orientations} of the edges, i.e., allocations such that each good is assigned to one of the two agents that values it positively. \citet{christodoulou2023fair} showed that EFX orientations may not exist, and that deciding their existence is an NP-complete problem. We strengthen their results to \efkx orientations, for any meaningful value of $k$.

\begin{inftheorem}\label{infthm:orientations}
\efkx orientations on graphs do not always exist. Furthermore, given an instance of the problem, deciding whether it admits an \efkx orientation is NP-complete.
\end{inftheorem}

\noindent Next we provide an overview of our techniques.

\subsection{An overview of our techniques}
\label{app:overview}

First, as a warm-up, in \cref{sec:efkx-approximations-easy}, we present a rather simple algorithm which achieves the weaker guarantee of $\frac{k}{k+1}$-\efkx. Then, for our main result in \cref{sec:efkx-approximations},  we generalize the approach taken in \citep{amanatidis2024pushing}. The \aefkx allocation in \cref{infthm:main-theorem} is obtained in three steps.

First, we obtain a partial allocation $\X^1$ which satisfies a certain set of desirable properties (see \cref{prop:a,prop:b,prop:c,prop:d,prop:e,prop:f} in \cref{sec:G3PA}); the most important of those properties are that the set of unallocated goods (called ``the pool'') contains a limited number of highly-valued (or ``critical'') goods for each agent, and only for agents whose bundles are singletons. This is achieved by the \emph{Generalized Property Preserving Allocation algorithm (\GPPA)}, an appropriate generalization of the \PPA algorithm of \citet{amanatidis2024pushing}, which was designed for the case of $k=1$. Compared to the original design, the algorithm works for any value of $k \geq 1$ and it is also conceptually simpler, as it uses only one type of envy graph (coined the ``modified graph'') rather than the two variants (the ``reduced'' and ``enhanced'' envy graphs) used in \citep{amanatidis2024pushing}, as well as fewer subroutines. The use of the modified graph also allows us to prove the polynomial running time of the algorithm, correcting an oversight of \citet{amanatidis2024pushing} which was identified by \citet{hv2025almost} (see \cref{rem:amanatidis-et-al-proof} in \cref{app:alternative-running-time-proof}). The proof of correctness of the algorithm follows closely that of the original version, with the necessary technical adaptations to make it work for the more general case of any $k$.

In the second step, we devise an algorithm which inputs the partial allocation $\X^1$ and produces a partial allocation $\X^2$, which is \aefkx, and for which the pool does not have any critical goods for any agent. In \citep{amanatidis2024pushing}, this was only achieved for special cases of the problem, e.g., for instances with up to $7$ agents. For $k \geq 2$, our novel algorithm, coined \AEC, achieves this allocation for any number of agents. The crucial insight that allows us to establish this result is that, after allocating a number of critical goods to an agent, the remaining goods in the pool are no longer critical for the agent, since her utility in the meantime has increased. This approach falls short for $k=1$; in this case we still manage to strengthen the state-of-the-art results from $7$ to $8$ agents. While the number of agents increases by only one, our newly introduced algorithm for allocating the critical goods is much more involved compared to the corresponding one in \citep{amanatidis2024pushing}, and handles the critical goods differently depending on the number of \emph{contested} critical goods and the number of sources in the modified envy graph.

Finally, in the third step, we use the Envy Cycle Elimination algorithm of \citep{lipton2004approximately} to transform the partial allocation $\X^2$ to a full allocation $\tilde{X}$ which is \aefkx. This is enabled by a crucial lemma, which establishes that this is possible, assuming the absence of any critical goods from the pool. This lemma was first formally stated and proved for $k=1$ by \citet{markakis2023improv}, but it was used implicitly earlier by \citet{ANM2019,farhadi2021almost}. Our version is an appropriate generalization for any value of $k$; we refer to this lemma as the \emph{Partial-to-Full-Allocation Lemma (P2FA)}, see \cref{lem:partial-efx-to-full-efx}.  \medskip

\noindent We consider \efkx graph orientations in \cref{sec:efkx-orientations}. First, we construct an instance in which \efkx orientations do not exist for any meaningful value of $k$, and then prove the corresponding NP-completeness result. Our NP-hardness result proceeds via induction on $k$, and using the NP-hardness result of \citet{christodoulou2023fair} for $k=1$ as the base case.

\subsection{Further Related Work}

On the topic of orientations, further to the works mentioned above, \citet{deligkas2024ef1} showed that EF1 orientations always
exist for monotone valuations and gave a pseudopolynomial time
algorithm for finding them. They also proved that deciding
whether an EFX orientation exists is NP-complete even for
restricted classes of graphs. \citet{zeng2024structure} provided characterizations for the existence of EFX orientations, and \citet{afshinmehr2024efx} studied EFX orientations in (bipartite) multigraphs.

Besides the aforementioned relaxations of $\alpha$-EFX and \efkx, the literature has studied further relaxations such as (approximate)-EFX with charity \citep{caragiannis2019charity,CKMS20,aaai/BergerCFF22}, EFL \citep{barman2018groupwise}, EFR \citep{farhadi2021almost}, EEFX and MXS \citep{caragiannis2023new}, and their combinations \citep{ashuri2025simultaneously}. We refer the reader to the survey of \citet{amanatidis2023fair} for a more detailed exposition.

\section{Preliminaries}\label{sec:preliminaries}

We consider the problem of fairly allocating a set $M$ of $m$ indivisible goods to a set $N$ of $n$ agents. Each agent $i$ has preferences over the goods, which are expressed in terms of an \emph{(additive)} valuation function $v_i: M \rightarrow \mathbb{R}_{\geq 0}$ that assigns non-negative values to each good; the value of an agent for a set (or a \emph{bundle}) $X \subseteq M$ of goods is defined as $v_i(X)=\sum_{g \in X}v_i(g)$. A partial allocation $\X$ is defined as a tuple $\X=(X_1,\ldots,X_n)$, where $X_i$ denotes the bundle allocated to agent $i$, such that (a) for every $i \in N$, it holds that $X_i \subseteq M$ and that (b) for any $i,j \in N$ such that $i \neq j$, it holds that $X_i \cap X_j = \emptyset$. A partial allocation $\X$ is a \emph{full allocation}, or simply, an \emph{allocation}, if additionally, it holds that $\bigcup_{i \in N}X_i = M$. For a partial allocation $\X$, the set $M\setminus \bigcup_{i \in N}X_i$ will be referred to as \emph{the pool}, and will be denoted by $\PP(\X)$. 

The fairness notion that we will be concerned with in this work is \emph{approximate envy-freeness up to $k$ goods}, which stipulates that an agent can only envy another agent's bundle up to a multiplicative factor $\alpha$, after the removal of any $k$ goods from the latter agent's bundle. We provide the formal definition below. 

\begin{definition}[$\alpha$-\efkx]
Given a partial allocation $\X$, a $k \in \{0,\ldots,m\}$, an $\alpha >0$, and two agents $i,j \in N$, we will say that agent $i$ is $\alpha$-\efkx towards agent $j$ if, for any $Y \subseteq X_j$ such that $|Y|=k$, it holds that $v_i(X_i) \geq \alpha \cdot v_i(X_j \setminus Y)$. We will say that $\X$ is $\alpha$-\efkx if every agent $i \in N$ is $\alpha$-\efkx towards any other agent $j \in N$.
\end{definition}
\noindent When $\alpha=1$, or $k=1$, these will be omitted from the notation, i.e., we will instead refer to \efkx (when $\alpha=1$), $\alpha$-EFX (when $k=1$), and EFX (when both $\alpha=1$, and $k=1$). When $k=0$, instead of $\alpha$-EF0X, we will say that an agent does not \emph{$\alpha$-envy} another agent, and $\X$ is $\alpha$-envy-free. Extending the definitions above, we will also need to compare the value of an agent for a bundle of goods $Y$ which is not necessarily the bundle of some other agent in $\X$. In that case, when $v_i(Y)\geq \alpha \cdot v_i(X_i)$, we will say that agent $i$ \emph{prefers $Y$ by a factor of at least $\alpha$}.

Following the terminology of \citet{amanatidis2024pushing}, we will refer to goods in $\PP(\X)$ that an agent $i$ values significantly compared to her current allocation as \emph{critical goods}.

\begin{definition}[$\beta$-critical good]
Given a partial allocation $\X$, a good $g \in \PP(\X)$ is \emph{$\beta$-critical} for agent $i \in N$, if $v_i(g) > \beta \cdot v_i(X_i)$.  
\end{definition}

\paragraph{Envy-graphs} A central concept in the literature of fair division with indivisible goods is that of an \emph{envy graph} \citep{lipton2004approximately}. Given a partial allocation $\X$, its \emph{envy graph} $G(\X)=(N,E(\X))$ is a directed graph that has one node for each agent, and a directed edge $(i,j)$ if and only if agent $i$ envies agent $j$, i.e., $v_i(X_j) > v_i(X_i)$. For the \PPA algorithm, \citet{amanatidis2024pushing}, defined two different variants of the envy graph, namely \emph{reduced envy graphs} and \emph{enhanced envy graphs}. Our generalization simplifies their approach and only makes use of a single variant, that we call the \emph{modified envy graph}.

\begin{definition}[Modified Envy Graph $\alpha$-$\Ge$]
Fix an integer \(k\ge 1\) and a parameter \(\alpha\in(0,1]\). Given a partial allocation \(\mathbf X=(X_i)_{i\in N}\) such that \(|X_i|\in\{1,k+1\}\) for every agent \(i\in N\), its modified $\alpha$-envy graph $\alpha$-$\Ge(\X)=(N,E(\X))$ is a directed graph that has one node for each agent, and a directed edge $(i,j)$ if and only if:
\begin{itemize}
    \item[-] $|X_i|=|X_j|$ and $v_i(X_j)>v_i(X_i)$
    \hfill \emph{(exact) envy}

    \item[-] $|X_i|=1$ and $|X_j|>1$ and
    $v_i(X_j)>\alpha\cdot v_i(X_i)$
    \hfill \emph{$\alpha$-approximate (weak) envy}

    \item[-] $|X_i|>1$ and $|X_j|=1$ and
    $v_i(X_j)>\frac{1}{\alpha}\cdot v_i(X_i)$
    \hfill \emph{$\frac{1}{\alpha}$-approximate (strong) envy}
\end{itemize}
An alternative but useful definition of the $\alpha$-$\Gr$ graph can be obtained via the following \emph{proxy valuation function}:
\[
\tilde{v}_i(Y)=
\begin{cases} 
v_i(Y), & |Y| = 1, \\
\frac{1}{\alpha} v_i(Y), & |Y| > 1.
\end{cases}
\]
Then, the modified envy graph $\Gr(\X)$ is simply the envy graph $G(\X)$ on an instance with valuation functions $\tilde{v}_i(\cdot)$ rather than $v_i(\cdot)$, for any $i \in N$.
\end{definition}
\noindent When it is clear from context, we will drop the partial allocation $\X$ from the notation and simply write $G$, and $\alpha$-$\Gr$ for the envy graph and the modified envy graph, respectively. Notice also that $G$ is simply $\alpha$-$\Ge$, when $\alpha=1$. Since $\alpha$ will be fixed for all of our results where $\alpha$-$\Ge$ is used (as a function of $k$), we will refer to the modified envy-graph $\Gr$ (dropping the $\alpha$) instead. 

 \subsection{Useful Subroutines}

Our algorithm will also make use of four useful subroutines, variants of which were defined in \citep[Section 2.3]{amanatidis2024pushing}. These apply to the graph variants defined above and exchange bundles along paths or cycles of the graph. The first three, namely $\CR$, $\ACR$ and $\PR$ 
are identical to those presented in \citep{amanatidis2024pushing}; we describe them below, and refer the reader to \citep[Section 2.3]{amanatidis2024pushing} (or, alternatively, \cref{app:section2-subroutines} of our paper) for the detailed pseudocode. The fourth one, $\PRPA$ is an appropriate generalization of the corresponding subroutine of the same name in \citep{amanatidis2024pushing}, which we present in full detail afterwards. Below, $\tilde{G}$ denotes any of the graph variants generically. 
\begin{itemize}[leftmargin=*]
    \item[] \textbf{Subroutine 1 -  $\CR(\X,\tilde{G},C)$:} For every edge $(i,j)$ in the cycle $C$ of $\tilde{G}$, assign the bundle $X_j$ of agent $j$ to agent $i$.
    \item[] \textbf{Subroutine 2 - $\ACR(\X,\tilde{G})$:} While there exists a cycle $C$ in $\tilde{G}$, apply \\ $\CR(\X,\tilde{G},C)$. 
    \item[] \textbf{Subroutine 3 - $\PR(\X,\tilde{G},\Pi)$:} Let $\Pi = (i_1,i_2,\ldots,i_\ell)$ be a path of $\tilde{G}(\X)$. For any $h \in \{1,\ldots,\ell-1\}$, assign the bundle of agent $i_{h+1}$ to agent $i_h$. 
\end{itemize}
\noindent Notice that the outcome of $\PR$ is not a partial allocation, as it does not specify what happens with the bundle of the last agent $i_\ell$ of the path; this will be determined by the algorithms that use $\PR$ as a subroutine. We conclude with the definition of the \PRPA subroutine, see below.

\SetAlgorithmName{SUBROUTINE}{Subroutine}{List of Subroutines}

\begin{algorithm}[!ht]
\DontPrintSemicolon
\caption{$\PRPA(\X,\tilde{G},\Pi,Y)$} 
\label{alg:path-resolution-and-critical}
\SetKwComment{Comment}{/* }{ */}
\KwData{A partial allocation $\X$, its graph $\tilde{G}(\X)$, a path $\Pi = (s, \ldots, i)$ in\, $\tilde{G}(\X)$ starting at a source $s$ of $\tilde{G}(\X)$ with $|X_s|=k+1$ and ending at a node $i$, and a set of goods $Y \subseteq X_s \cup \PP(\X)$, with $|Y| = k+1$.}
\KwResult{An updated partial allocation $\X$.\vspace{3pt}}
$(X_j)_{j \in N: \exists (j, \ell) \in \Pi}=\PR(\X,\tilde{G},\Pi)$\;
$X_i \gets Y$\; 
$\PP(\X)\gets M\setminus\bigcup_{j\in N}X_j$\;
\Return $\X$
\end{algorithm}

The main difference between Subroutine~\ref{alg:path-resolution-and-critical} and the corresponding subroutine in \citep{amanatidis2024pushing} is in the bundle that agent $i$ receives. In their work, agent $i$ receives, from her perspective, the best item out of the two goods that constitute the bundle of the source $X_s$, and the best good from the pool $\PP(\X)$. In our case of $k \geq 1$, the source $s$ will have $k+1$ goods, and it will not be a-priori clear which of those agent $i$ should receive together with which goods from $\PP(\X)$. The subroutine thus receives the bundle $Y$ consisting of these goods as input, and delegates this decision to the main algorithm that calls it; in our case that is the \GPPA algorithm of \cref{sec:efkx-approximations}. 

Finally, we present the \emph{Envy Cycle Elimination} algorithm of \citet{lipton2004approximately}, adapted (for our purposes) to start from any partial allocation $\X$. As long as $\PP(\X) \neq \emptyset$, the algorithm repeatedly finds a source of the envy graph $G(\X)$ and allocates an arbitrary good $g \in \PP(\X)$ to that source. If no such source exists in $G(\X)$, that means that $G(\X)$ might have at least one cycle. The algorithm then ``resolves'' that cycle using the \CR subroutine described above; we refer to this algorithm as $\ECE(\X,G)$, and present its pseudocode for completeness in \cref{app:section2-subroutines}.

\subsection{Partial to Full Allocation (P2FA) Lemma}

At the heart of our approach lies a simple but useful lemma, which allows us to translate partial allocations that satisfy certain properties to full approximate \efkx allocations. Such lemmas were used (implicitly) in the past by other works \citep{ANM2019,farhadi2021almost}), but the first explicit statement appeared in \citep{markakis2023improv}, see also \citep[Lemma 4.2]{amanatidis2024pushing}. The statement in \citep{markakis2023improv} only concerns EFX allocations (rather than \efkx); we state the appropriate generalization below. Its proof is a straightforward adaptation, but we include it below for completeness.

\begin{lemma}[Partial-to-Full-Allocation (P2FA) Lemma]\label{lem:partial-efx-to-full-efx}
Let $\alpha\in(0,1]$ and $\beta\in[0,1]$. Let $\X$ be an $\alpha$-\efkx partial allocation such that $\PP(\X)$ does not contain any $\beta$-critical good for any agent. Then a $\min\{\alpha,\frac{1}{\beta+1}\}$-\efkx allocation $\tilde{\X}$ can be obtained from $\X$ in polynomial time.
\end{lemma}

\begin{proof}
    The full allocation $\tilde{\X}$ will be obtained from $\X$ by allocating the goods from the pool $\PP(\X)$ using the $\ECE$ algorithm (see Algorithm~\ref{alg:ece} in \cref{app:section2-subroutines}). First, from the definition of $\ECE$, it follows that, during the run of the algorithm, the utility of an agent $i \in N$ never decreases (see \citep{lipton2004approximately} and \citep[Lemma 4.1]{amanatidis2024pushing}). Consider any agent $i \in N$ and any agent $j \in N \setminus \{i\}$. If suffices to prove that in $\tilde{X}$, agent $i$ is $\min\{\alpha,\frac{1}{\beta+1}\}$-\efkx towards agent $j$. We consider two cases:
    \begin{itemize}[leftmargin=*]
        \item[-] \emph{Case 1:} $\tilde{X}_j = X_{j'}$, i.e., the final bundle of agent $j$ in allocation $\tilde{\X}$ was the exact bundle of some (possibly other) agent $j'$ in partial allocation $\X$, without adding any items from $\PP(\X)$. Then, since $v_i(\tilde{X}_i) \geq v_i(X_i)$ as discussed above, and since agent $i$ was $\alpha$-\efkx towards agent $j'$ in $\X$, she will now be $\alpha$-\efkx towards agent $j$ in $\tilde{X}$. 
        \item[-] \emph{Case 2:} One or more goods from the initial pool $\PP(\X)$ have been added to the bundle $\tilde{X}_j$ during the execution of the algorithm. Let $g$ be the last such good, and let $j'$ be the agent holding this bundle immediately before $g$ was added.
        Let $\tilde{\X}^t$ denote the allocation at that time; thus, $\tilde{X}_{j'}^t=\tilde{X}_j\setminus\{g\}$. By definition of the algorithm, agent $j'$ was a source of $G(\tilde{\X}^t)$, and hence $v_i(\tilde{X}_i^t) \geq v_i(\tilde{X}_{j'}^t)=v_i(\tilde{X}_j \setminus \{g\})$. Since $g$ was not $\beta$-critical for agent $i$ in the input allocation, we have $v_i(g)\leq\beta v_i(X_i) \leq\beta v_i(\tilde{X}_i^t)$. Therefore, we have that:
        \[
        \begin{aligned}
        v_i(\tilde{X}_j) &= v_i(\tilde{X}_j \setminus \{g\}) + v_i(g) \\ &\leq (1+\beta) \cdot v_i(\tilde{X}_i^t) \leq (1+\beta)\cdot  v_i(\tilde{X}_i). 
        \end{aligned}
        \]
        This establishes that agent $i$ is $\frac{1}{\beta+1}$-envy-free (and hence $\frac{1}{\beta+1}$-\efkx) towards agent $j$. This completes the proof. 
    \end{itemize}
    
\end{proof}


    
    

\section{Approximate \efkx allocations}\label{sec:approx-efx}

\subsection{Warm-up: $\frac{k}{k+1}$-\efkx allocations} \label{sec:efkx-approximations-easy}

Before we present our main algorithm which achieves $\frac{k+1}{k+2}$-\efkx allocations in \cref{sec:efkx-approximations} below, we first observe that a rather straightforward polynomial-time algorithm suffices to achieve the weaker \efkx approximation of $\frac{k}{k+1}$. The algorithm first constructs a partial allocation $\X$, simply by $k$ consecutive executions of the \emph{Round-Robin} algorithm, which fixes an ordering of the agents $i \in N$, and lets them pick their favorite good from $\PP(\X)$ according to that ordering. Then, the algorithm completes $\X$ by running the Envy Cycle Elimination algorithm of \cite{lipton2004approximately} for the leftover goods in $\PP(\X)$. We refer to this algorithm as \textsc{$k$-RoundRobin with EnvyCycleElimination} and present its pseudocode, as well as the brief proof for its \efkx approximation guarantee below.

\setcounter{algocf}{1}
\SetAlgorithmName{ALGORITHM}{Algorithm}{List of Algorithms}

\begin{algorithm}[!ht]
\DontPrintSemicolon
\caption{\textsc{$k$-RoundRobin with EnvyCycleElimination}}
\label{alg:rr-efkx}
\SetKwComment{Comment}{/* }{ */}
\KwData{Valuations $(v_i(g))_{i \in N,\, g \in M}$, $k\geq 1$.}
\KwResult{A $\frac{k}{k+1}$-\efkx allocation $\X$.\vspace{3pt}}

$\X\gets(\emptyset,\ldots,\emptyset)$\;
Fix an ordering $I=(i_1,\ldots,i_n)$ of the agents\;

\For{$r\gets1$\ \textnormal{to}\ $k$}{
    \For{$i_\ell\in I$}{
        \If{$\PP(\X)=\emptyset$}{
            \Return $\X$\;
        }
        Pick $g'\in\arg\max_{g\in\PP(\X)}v_{i_\ell}(g)$\;
        $X_{i_\ell}\gets X_{i_\ell}\cup\{g'\}$\;
        $\PP(\X)\gets\PP(\X)\setminus\{g'\}$\;
    }
}

$\X = \ECE(\X,G(\X))$\;

\Return $\X$
\end{algorithm}

\begin{theorem}\label{thm:efkx-approx-easy}
    For any $k \in [1,n]$, the \textsc{$k$-RoundRobin with EnvyCycleElimination} algorithm runs in polynomial time, and achieves a $\frac{k}{k+1}$-\efkx allocation.
\end{theorem}

\begin{proof}[Proof of \cref{thm:efkx-approx-easy}]
    The polynomial running time of the algorithm follows immediately from the fact that both Round Robin and Envy Cycle Elimination are known to run in polynomial time (e.g., see \citep{amanatidis2023fair}). For the approximation to \efkx, by \cref{lem:partial-efx-to-full-efx}, it suffices to show that the partial allocation $\X$ that is produced after the $k$ rounds of Round Robin is is $\alpha$-\efkx, and that $\PP(\X)$ does not contain any $\beta$-critical goods, for $\alpha=k/(k+1)$ and $\beta=1/k$. For the former, it is easy to see that $\X$ is in fact \efkx, as for every agent $i \in N$, $|X_i| \leq k$, and hence every agent is \efkx towards every other agent. For the latter, consider any agent $i \in N$ with bundle $X_i$ in $\X$. Since every good in $X_i$ was picked via Round Robin, this implies that for every pair of goods $(g',g)$ such that $g' \in \PP(\X)$ and $g \in X_i$, it holds that $v_i(g') \leq v_i(g)$. By summing over items in $X_i$, we obtain that $k\cdot v_i(g') \leq \sum_{g \in X_i}v_i(g) = v_i(X_i)$; this latter inequality precisely establishes that $g'$ is not $(1/k)$-critical for agent $i$.   
\end{proof}

\begin{figure*}[t]
\centering
\begin{tcolorbox}[
    standard jigsaw,
    title=Desired Properties of a Partial Allocation $\X$:,
    opacityback=0,  
]
\begin{properties}[topsep=5pt,itemsep=0.5ex,partopsep=1ex,parsep=1ex]
\item For every agent $i \in N$, we have either $|X_i|=1$ or $|X_i|=k+1$. \label{prop:a}
\item Every agent $i \in N$ with $|X_i|=1$ is \efkx towards any other agent.  
\label{prop:b}
\item Every agent $i \in N$ is $\frac{k+1}{k+2}$-\efkx towards any other agent. 
\label{prop:c}
\item Every agent $i$ values her own bundle at least as much as any single unallocated good, i.e., for every agent $i \in N$ and good $g \in \PP(\X)$, $v_i(X_i)\geq v_i(g)$.\label{prop:d}
\item Every agent $i \in N$ with $|X_i|=k+1$ does not have any $\frac{1}{k+1}$-critical goods, i.e., for any good $g \in \PP(\X)$, we have that $v_i(g)\leq \frac{1}{k+1} v_i(X_i)$.\label{prop:e}
\item For every agent $i\in N$ with $|X_i|=1$, let
$Y_i^c=\{g\in\PP(\X):v_i(g)>\frac{1}{k+1}v_i(X_i)\}$
be the set of her $\frac{1}{k+1}$-critical goods. Then
$|Y_i^c|\leq k$ and
$v_i(Y_i^c)\leq\frac{k+1}{k+2}v_i(X_i)$.\label{prop:f}
\end{properties}
\end{tcolorbox}
\end{figure*}

\subsection{$\frac{k+1}{k+2}$-\efkx allocations} \label{sec:efkx-approximations}

We now move on to our main positive result. For every fixed
$k\geq 2$, there is a polynomial-time algorithm that achieves
the improved $(k+1)/(k+2)$ approximation bound on \efkx allocations. The algorithm consists of the following three parts:
\begin{enumerate}[leftmargin=*]
    \item For every fixed $k\geq 1$, a polynomial-time algorithm that either produces a full \aefkx allocation, or a partial allocation $\X^1$, that is $\frac{k+1}{k+2}$-\efkx and satisfies a set of desirable properties, see \cref{prop:a,prop:b,prop:c,prop:d,prop:e,prop:f}. $\PP(\X^1)$ may still have $\frac{1}{k+1}$-critical goods, but these properties apply certain restrictions on them, which will be utilized in the next step. We present this algorithm in \cref{sec:G3PA}.
    \item A polynomial-time algorithm which, for any $k \geq 2$, inputs the partial allocation $\X^1$ and outputs a partial allocation $\X^2$, which is still $\frac{k+1}{k+2}$-\efkx, but in which $\PP(\X^2)$ does not have any $\frac{1}{k+1}$-critical goods. We present this algorithm in \cref{sec:completing}.
    \item By the P2FA Lemma (\cref{lem:partial-efx-to-full-efx}), a polynomial-time algorithm which inputs $\X^2$ and outputs a full $\frac{k+1}{k+2}$-\efkx allocation $\tilde{\X}$.
\end{enumerate}

\noindent The pseudocode for the algorithm, which we refer to as \AEFKX, is presented in Algorithm~\ref{alg:approx-efkx}.

\setcounter{algocf}{2}
\SetAlgorithmName{ALGORITHM}{Algorithm}{List of Algorithms}

\begin{algorithm*}[H]
\DontPrintSemicolon
\caption{\AEFKX}
\label{alg:approx-efkx}
\SetKwComment{Comment}{/* }{ */}
\KwData{Valuations $(v_i(g))_{i \in N,\, g \in M}$, $k\geq 2$.}
\KwResult{A \aefkx allocation $\X$.\vspace{3pt}\vspace{3pt}}

Let $\X^0$ be an allocation where each bundle has size $1$\;
\tcp{Start with an initial seed allocation (\cref{def:seed-allocation}) $X^0$.}
$\X^1 \gets \GPPA((v_i(g))_{i \in N,\, g \in M}, \X^0, k)$ \;
\tcp{Run the \GPPA algorithm (Algorithm~\ref{alg:G3PA}) on $\X^0$ to obtain a partial allocation $\X^1$ satisfying \cref{prop:a,prop:b,prop:c,prop:d,prop:e,prop:f} (\cref{sec:G3PA}).}
$\X^2 \gets \AEC((v_i(g))_{i \in N,\, g \in M}, \X^1, k)$\;
\tcp{Run the \AEC algorithm (Algorithm~\ref{alg:aec}) on $\X^1$ to obtain a partial allocation $\X^2$ that is \aefkx and does not have any $\frac{1}{k+1}$-critical goods (\cref{sec:completing}).}
$\X\gets\ECE(\X^2,G(\X^2))$\;
\tcp{Run the \ECE algorithm (Algorithm~\ref{alg:ece}) to obtain a full \aefkx allocation (\cref{lem:partial-efx-to-full-efx}).}
\Return $\X$
\end{algorithm*}

\subsubsection{The Generalized Property Preserving Partial Allocation (\GPPA) algorithm}\label{sec:G3PA}

For the first part, we devise an algorithm which is an appropriate generalization of the \emph{Property Preserving Partial Allocation (\PPA)} algorithm of \citet{amanatidis2024pushing}; we refer to the algorithm as the \emph{Generalized Property Preserving Partial Allocation (\GPPA) algorithm}. Similarly to \PPA, \GPPA performs a series of steps, each of which may be executed only if the previous one has not. These steps are either \emph{swap steps}, where the agents are allowed to swap (parts of) their bundles with goods from the pool, or \emph{graph reallocation steps}, which operate on the graph $\Ge$, and reallocate bundles along paths or cycles of the graph, possibly also adding some goods from the pool or releasing some good to it.

The algorithm starts with an initial allocation (referred to as a \emph{seed allocation}, see \cref{def:seed-allocation} below) $\X$ and updates $\X$ to a partial allocation that satisfies a certain set of properties. These properties, which are appropriate generalizations of the properties in \citep{amanatidis2024pushing}, are listed above.

While in \PPA, the partial allocation $\X$ contains bundles with at most two goods, for our updated set of properties, the bundle size should be at most $k+1$; in fact, we prove that it suffices to create bundles of size either $1$ or $k+1$. In turn, any agent with a bundle of size $k+1$ has no
$1/(k+1)$-critical goods, and any agent with a singleton bundle
has at most $k$ $1/(k+1)$-critical goods, which she values in total
at most $(k+1)/(k+2)$ times the value of her bundle in $\X$.

\begin{definition}[Seed allocation]\label[definition]{def:seed-allocation}
A \emph{seed allocation} $\X$ is a partial allocation that satisfies \cref{prop:a,prop:b,prop:c} above.
\end{definition}

\noindent As in \citep{amanatidis2024pushing}, initializing the algorithm with a seed allocation is done for flexibility in its use on partial allocations. For all of our results, taking an arbitrary partial allocation $\X$ in which $|X_i|=1$ for every agent $i \in N$ would suffice. \medskip

\begin{algorithm}[!ht]
\setstretch{1.05}
\DontPrintSemicolon
\SetNoFillComment
\LinesNotNumbered 
\caption{\textsc{\sc Generalized Property-Preserving Partial Allocation (\GPPA)} 
$\left(\left(v_i\right)_{i 
\in N}, \X, k \right)$} \label{alg:G3PA}
\SetKwComment{Comment}{/* }{ */}
\SetKw{Continue}{continue}
\SetKw{Break}{break}
\SetKw{Step}{Step}
\SetKw{Substep}{Substep}
\SetKwData{Kw}{}
\KwData{The values $v_i(g)$ for every $i \in N$ and every $g \in M$, a \emph{seed allocation} (\cref{def:seed-allocation}, and $k\geq 1$). 
}
\KwResult{A partial allocation $\X^1$ which satisfies \Cref{prop:a,prop:b,prop:c,prop:d,prop:e,prop:f}.\vspace{3pt}}
\While{$\PP(\X) \neq \emptyset$}
{
\nl \Kw{\texttt{Step 1}} \label{step1}
\uIf{there is $i \in N$ with $|X_i|=1$ and a good $g \in \PP(\X)$ such that $v_i(g) > v_i(X_i)$}{
$\PP(\X) \gets (\PP(\X)\cup X_i)\setminus \{g\}$ and $X_i \gets \{g\}$ \;
\tcp{If an agent with 1 good prefers 1 good from the pool, swap them.}}

\nl \Kw{\texttt{Step 2}} \label{step2}
\uElseIf{there is $i \in N$ with $|X_i|=k+1$ and a good $g \in \PP(\X)$ such that $v_i(g) > \frac{k+2}{k+1}v_i(X_i)$}{
$\PP(\X) \gets (\PP(\X)\cup X_i)\setminus \{g\}$ and $X_i \gets \{g\}$ \;
\tcp{Else if an agent with $k+1$ goods prefers 1 good from the pool by more than $(k+2)/(k+1)$, swap her bundle with that good.}
}

\nl \Kw{\texttt{Step 3}} \label{step3}
\uElseIf{there is $i \in N$ with $|X_i|=1$ and $k+1$ goods $g_1,\ldots,g_{k+1} \in \PP(\X)$ such that $v_i(\{g_1,\ldots,g_{k+1}\}) > \frac{k+1}{k+2}v_i(X_i)$}{
$\PP(\X) \gets (\PP(\X)\cup X_i)\setminus \{g_1,\ldots,g_{k+1}\}$ and $X_i \gets \{g_1,\ldots,g_{k+1}\}$\;
\tcp{Else if an agent with 1 good prefers $k+1$ goods from the pool by more than $(k+1)/(k+2)$, swap that 1 good with the $k+1$ goods from the pool.}
}

\nl \Kw{\texttt{Step 4}}\label{step4}
\uElseIf{there is $i \in N$ with $|X_i|=k+1$, and goods $g \in \PP(\X)$ and $g' \in X_i$ such that $v_i(g) > v_i(g')$}
{ $\PP(\X) \gets (\PP(\X) \cup \{g'\}) \setminus \{g\}$ and $X_i \gets (X_i \cup \{g\}) \setminus \{g'\}$.\;
\tcp{Else if an agent with $k+1$ goods prefers 1 good from the pool to one of her own goods, swap that good with the one good from the pool.}
}

\nl \Kw{\texttt{Step 5}}\label{step5}
\uElseIf{the modified envy graph $\Gr(\X)$ has cycles}{$\X \gets\ACR(\X,\Gr)$\;
\tcp{Else if the modified envy graph has cycles, resolve them by swapping the bundles.}}

\nl \Kw{\texttt{Step 6}}\label{step6}
\uElseIf{there exists a source $s$ in the modified envy graph $\Ge(\X)$ with $|X_s| = 1$}{
     \tcp{If there is a source in the modified envy graph with a single good,}
    
    \nl \Kw{\texttt{Substep 6.1}}\label{step61}
   \uIf{$|\mathcal{P}(\mathbf{X})| \ge k$}{
        Let $Y^{*} \in \arg\max_{Y \subseteq \PP(\X): |Y| = k} v_s(Y)$\;
        $\mathcal{P}(\mathbf{X}) \gets \mathcal{P}(\mathbf{X}) \setminus Y^{*}$\;
        $X_s \gets X_s \cup Y^{*}$\;
        \tcp{if there are at least $k$ goods in the pool, add the $k$ most valuable goods from the pool to the bundle of agent $s$.}
    }
    
    \nl \Kw{\texttt{Substep 6.2}}\label{step62}
    \uElse{
    $X_s \gets X_s \cup \PP(\X)$\;
    \Break\;
     \tcp{if there are at fewer than $k$ goods in the pool, add all of those goods to the bundle of agent $s$ and terminate.}
    }
}
\setcounter{AlgoLine}{6}
\nl \Kw{\texttt{Step 7}} \label{step7}
\uElseIf{there exists a source $s$ and a path $\Pi = (s,\ldots,i)$ in the modified graph $\Ge(\X)$ from $s$ to some agent $i$ such that $|X_i|=1$, and there exists a set of goods $Y \subseteq \PP(\X) \cup X_s$ with $|Y|=k+1$, such that $v_i(Y) > \frac{k+1}{k+2}\cdot v_i(X_i)$}
{$\X \gets \PRPA(\X,\Ge(\X),\Pi,Y)$\;
\tcp{Else there exists a source $s$ in the modified graph, and a path to some agent $i$ that has a singleton bundle, use the \PRPA subroutine to update the allocation, assuming that agent $i$ prefers an appropriate bundle $Y$ of $k+1$ goods from $X_s$ and the pool to her own bundle by a factor of more than $(k+1)/(k+2)$.}}


\nl \Kw{\texttt{Step 8}}\label{step8}
\Else {\Break\;}
}
\Return $\X$ 
\end{algorithm}

\noindent The \GPPA algorithm is presented in Algorithm~\ref{alg:G3PA}. The following lemma establishes its properties. Its proof is provided in \cref{app:proofG3PA}.

\begin{lemma}\label{lem:GPA-satisfies-properties}
Let $\X$ be a seed allocation, and let $\X^1 = \GPPA((v_i)_{i \in N},\X, k)$ be the outcome of the \GPPA algorithm on input $\X$. Then either (i) $\X^1$ is a $\frac{k+1}{k+2}$-\efkx full allocation or (ii) $\X^1$ satisfies \cref{prop:a,prop:b,prop:c,prop:d,prop:e,prop:f}, and, furthermore, the modified graph $\Ge(\X^1)$ has at least one source, and every source $s$ in $\Ge(\X^1)$ has $|X_s^1|=k+1$. 
\end{lemma}

\noindent Before we conclude the subsection, we prove that the \GPPA algorithm runs in polynomial time, for any fixed $k$. This follows along the same lines of the corresponding proof in \citep[Lemma 3.5.]{amanatidis2024pushing} for the \PPA algorithm. We provide the proof below for completeness.\footnote{The proof as stated in \citep{amanatidis2024pushing} contains an oversight, which was identified by \cite{hv2025almost}, see \cref{rem:amanatidis-et-al-proof} in \cref{app:alternative-running-time-proof}. Our use of the modified graph $\Gr$ also corrects this oversight.}

\begin{theorem}\label{thm:running-time-G3PA}
    Given any fixed $k$, the \GPPA algorithm runs in polynomial time.
\end{theorem}

\begin{proof}[Proof of \cref{thm:running-time-G3PA}]
Before the execution of Substep~6.2, whenever agent
$i$'s bundle changes, $\tilde{v}_i(X_i)$ strictly increases,
except possibly in Substep~6.1. In that substep, the proxy
utility does not decrease and the bundle size increases from
$1$ to $k+1$. Any subsequent change back to a singleton
strictly increases the proxy utility. Therefore, no agent
can return to a bundle she previously held.

This means that during the execution of the algorithm, there
are at most $\binom{m}{k+1}+\binom{m}{1}+1$ distinct bundles
that an agent can receive; this is because, before
Substep~6.2, \cref{prop:a} holds and hence each bundle has
size either $1$ or $k+1$. The last term in the sum above
corresponds to the only bundle of size
$\ell\notin\{1,k+1\}$ that an agent can be assigned in
Substep~6.2, when $\X^1$ is a full allocation.
For fixed $k$, this quantity is $O(m^{k+1})$.
Since every nonterminal iteration changes at least one
agent's bundle, the algorithm terminates after at most
$O(nm^{k+1})$ iterations of the while loop.
Therefore, it suffices to show that each step of the
algorithm is executed in polynomial time. This is
straightforward to see for all steps, except
Step~\ref{step5}, which invokes the \ACR subroutine on
$\Gr$. Each \CR subroutine in \ACR clearly runs in
polynomial time, so it suffices to bound the number of calls
to \CR within \ACR. For that, we will prove that in each
such call, the number of edges in $\Gr$ strictly decreases,
and hence the number of calls is bounded by $n^2$.

To see this, recall the alternative definition of the
modified graph $\Gr$, by means of the proxy valuation
function $\tilde{v}_i(\cdot)$. An edge $(i,j)$ in the graph
exists if and only if
$\tilde{v}_i(X_j)>\tilde{v}_i(X_i)$.
A cycle rotation preserves the multiset of allocated bundles,
and every participating agent strictly increases the proxy
utility of her own bundle. By these observations, every time
\CR is called, the number of outgoing edges for agent $i$
strictly decreases if agent $i$ is involved in the cycle.
Otherwise, the number of outgoing edges remains the same.
Since at least one agent is involved in every call of \CR,
the statement follows. This completes the proof.
\end{proof}

\subsubsection{Allocating the Critical Goods}\label{sec:completing}

We now move on to the second part of the algorithm, which establishes how to allocate the $\frac{1}{k+1}$-critical goods in $\PP(\X^1)$ in order to obtain the partial allocation $\X^2$. Contrary to the results of \citet{amanatidis2024pushing}, here we obtain that, for any $k\geq 2$, \aefkx allocations exist for any number of agents, without any restrictions on the valuation functions.

At first glance, the case of $k \geq 2$ seems equally challenging: although we now have a more relaxed fairness notion, $\PP(\X^1)$ also has more critical goods per agent (up to $k$), which need to be allocated. The crucial insight which enables the unconditional bound is the following: once we allocate some of the critical goods to an agent, the remaining goods are no longer critical, as her allocation has been updated and her utility has increased. This is established via Algorithm~\ref{alg:aec}, which we refer to as \AEC. \medskip

\setcounter{algocf}{4}
\SetAlgorithmName{ALGORITHM}{Algorithm}{List of Algorithms}

\begin{algorithm*}[!ht]
\setstretch{1.05}
\DontPrintSemicolon
\SetNoFillComment
\LinesNumbered
\caption{\textsc{\AEC}
$\left(\left(v_i\right)_{i \in N}, \X, k \right)$}
\label{alg:aec}

\SetKwComment{Comment}{/* }{ */}

\KwData{The values $v_i(g)$ for every $i \in N$ and every $g \in M$, a partial allocation $\X$ which satisfies \cref{prop:a,prop:b,prop:c,prop:d,prop:e,prop:f}, and $k\geq 2$.}
\KwResult{A partial allocation $\X^2$ which is
$\frac{k+1}{k+2}$-\efkx and in which $\PP(\X^2)$ does not contain
any $\frac{1}{k+1}$-critical goods.}


\While{there exists an agent $i \in N$ who has a $\frac{1}{k+1}$-critical good $g \in \PP(\X)$}{
    \uIf{$|\PP(\X)| \ge k-1$}{
        Let $Y^{*}\in 
        \arg\max_{Y \subseteq \PP(\X): |Y|=k-1}
        v_i(Y)$\;
        $\PP(\X) \gets \PP(\X) \setminus Y^{*}$\;
        $X_i \gets X_i \cup Y^{*}$\; \label{step-aec-5}
        
        \tcp{If there are enough goods in the pool, allocate the most valuable $k-1$ of them to agent $i$.}
    }
    \Else{
        $X_i \gets X_i \cup \PP(\X)$\; \label{step-aec-7}
        $\PP(\X) \gets \emptyset$\; 
        \tcp{Otherwise, allocate all the remaining goods from the pool to agent $i$.}
    }
}

\Return $\X$\;
\end{algorithm*}

\noindent We first prove the following lemma.

\begin{lemma}\label{lem:criticals-eliminated}
Let $i \in N$ be an agent such that $X_i$ was augmented to $X_i'$ in any step of the \AEC algorithm. After that step, the agent no longer has any $\frac{1}{k+1}$-critical goods in the pool $\PP(\X)$. 
\end{lemma}

\begin{proof}
Consider the first time agent $i$'s bundle is augmented.
Before this step, $X_i$ is unchanged from the input allocation,
and the pool has only shrunk. Since the input partial allocation $\X$ to \AEC satisfies \cref{prop:a,prop:e}, it must be the case that $|X_i|=1$, as agents with $k+1$ goods do not have any $\frac{1}{k+1}$-critical goods in $\PP(\X)$. Additionally, by \cref{prop:f}, for the set of $\frac{1}{k+1}$-critical goods $Y_i^c$, we have that $|Y_i^c| \leq k$, and that $v_i(Y_i^c) \leq \frac{k+1}{k+2}\cdot v_i(X_i)$. Obviously, if $|Y_i^c| \leq k-1$, then the algorithm will allocate all of the goods in $Y_i^c$ to the agent in $\X'$(either in Step~\ref{step-aec-5} or Step~\ref{step-aec-7}), and $Y_i^c \cap \PP(\X') = \emptyset$. So it remains to consider the case when $|Y_i^c|=k$.

Let $Y_i^c=\{g_1,\ldots,g_k\}$ and assume without loss of generality that $v_i(g_\ell) \geq v_i(g_{\ell+1})$ for all $\ell \in \{1,\ldots, k-1\}$. We have
\begin{equation*}
v_i(g_k) \leq \frac{1}{k}v_i(Y_i^c) \leq \frac{k+1}{k(k+2)}\cdot v_i(X_i).
\end{equation*}
Furthermore, we have 
\begin{equation*}
v_i(Y_i^c\setminus\{g_k\})
\geq (k-1)\cdot v_i(g_{k-1})
> \frac{k-1}{k+1}\cdot v_i(X_i),
\end{equation*}
where the first inequality follows from the additivity of the valuation function (since $g_{k-1}$ is the least valuable good in $Y_i^c \setminus \{g_k\}$), and the second inequality follows from the fact that $g_{k-1}$ is $\frac{1}{k+1}$-critical, and hence $v_i(g_{k-1}) > \frac{1}{k+1}v_i(X_i)$. After the allocation $X_i$ gets updated to $X_i'$, the whole of $Y_i^c\setminus\{g_k\}$ is allocated to agent $i$ in Step~\ref{step-aec-5}, and we have
\begin{align*}
\frac{1}{k+1}v_i(X_i') &= \frac{1}{k+1}v_i(X_i \cup Y_i^c \setminus \{g_k\})\\ &\geq \frac{1}{k+1} \left(v_i(X_i) + \frac{k-1}{k+1}v_i(X_i)\right) \\
&= \frac{2k}{(k+1)^2}\cdot v_i(X_i) \\&\geq \frac{k+1}{k(k+2)} v_i(X_i) \geq v_i(g_k),
\end{align*}
where the second to last inequality follows from $k\geq 2$ and
$v_i(X_i)\geq 0$. This establishes that $g_k$ is no longer a critical good for agent $i$ in $\PP(\X')$. Goods that were not critical before this step remain noncritical,
since $v_i(X_i')\geq v_i(X_i)$. Since the pool only shrinks and agents' utilities never decrease,
agent $i$ cannot have a critical good in any subsequent step.
\end{proof}

\noindent We have the following theorem.

\begin{theorem}\label{thm:AEC}
Let $k\geq 2$, let $\X$ be a partial allocation that satisfies
\cref{prop:a,prop:b,prop:c,prop:d,prop:e,prop:f}, and let
$\tilde{\X}$ be the output of $\AEC((v_i)_{i\in N},\X,k)$.
Then $\tilde{\X}$ is a partial \aefkx allocation such that
$\PP(\tilde{\X})$ does not contain any
$\frac{1}{k+1}$-critical goods. Furthermore, the algorithm runs
in polynomial time.
\end{theorem}

\begin{proof}
By definition, as long as the algorithm terminates, there will not be any $\frac{1}{k+1}$-critical goods in the pool $\PP(\tilde{X})$. Termination follows from \cref{lem:criticals-eliminated}, as each agent's bundle is augmented at most once; after that the agent no longer has any $\frac{1}{k+1}$-critical goods, and therefore her bundle is not augmented again by the algorithm. It is also straightforward to see that each step of the algorithm runs in polynomial time. Therefore, it remains to argue that $\tilde{\X}$ is a \aefkx partial allocation. 
Consider any agent $i \in N$ and any other agent $j \in N \setminus \{i\}$. We consider two cases:
\begin{itemize}[leftmargin=*]
     \item[-] \emph{Case 1: $\tilde{X}_j \supset X_j$, i.e., agent $j$'s bundle was augmented during the execution of the algorithm.} In this case, by \cref{prop:e}, we have $|X_j|=1$. Therefore, since the algorithm does not augment bundles by more than $k-1$ goods, it follows that $|\tilde{X}_j|\leq k$. This implies that agent $i$ is \efkx towards agent $j$. 
     \item[-] \emph{Case 2: $\tilde{X}_j = X_j$, i.e., agent $j$'s bundle was not augmented during the execution of the algorithm.} By \cref{prop:c}, agent $i$ was \aefkx towards agent $j$ in $\X$. Furthermore, it holds that $v_i(\tilde{X}_i) \geq v_i(X_i)$, and therefore agent $i$ is \aefkx towards agent $j$ in $\tilde{\X}$ as well.
\end{itemize}
\end{proof}

\subsection{A $2/3$-EFX allocation for 8 agents}
Our results in the previous section obtain $\frac{k+1}{k+2}$-approximations to \efkx for any $k \geq 2$. As we noted in the Introduction, for $k=1$ and general additive valuations, \citet{amanatidis2024pushing} managed to get a $2/3$ approximation to EFX only where there are at most $7$ agents. We improve this result to the case of $8$ agents in the following theorem.

\begin{theorem}\label{thm:8-agents-23-efx}
There is a polynomial-time algorithm which computes $2/3$-EFX allocations for instances with at most $8$ agents. 
\end{theorem}

\noindent The proof of the theorem is included in \cref{app:8-agents}; here we provide some high-level intuition. The algorithm that we construct makes use of the \GPPA algorithm, or, to be precise, a slight modification that we refer to as \GPPAplus, see Algorithm~\ref{alg:G3PA+}. This variant introduces an extra step and still produces a partial allocation $\X^1$ which satisfies \cref{prop:a,prop:b,prop:c,prop:d,prop:e,prop:f}. This is completely analogous to the approach in \citep{amanatidis2024pushing} for their $7$-agent $2/3$-EFX algorithm, and \GPPAplus is a generalization of their \PPAplus algorithm.  

This extra step is used in order to argue that the partial allocation produces after allocating \emph{contested} critical goods is $2/3$-EFX. A contested critical good is a good that is critical for at least two agents at the same time. Once the contested critical goods are allocated, an algorithm called \UCG (see Algorithm~\ref{alg:ucc} in \cref{app:8-agents}) is called to allocate the remaining critical goods. Compared to the algorithm of \citet{amanatidis2024pushing}, the algorithm for allocating the contested critical goods in our case is notably more involved; we refer to this algorithm as \ACC, see Algorithm~\ref{alg:acc} in \cref{app:8-agents}. \ACC performs a different allocation depending on the number of contested critical goods and the number of sources in the modified envy graph. We present the details of the algorithms and their analysis in \cref{app:8-agents}.

\section{\efkx Orientations on Graphs}\label{sec:efkx-orientations}

In this section, we study \efkx orientations on graphs. We consider
valuation functions as defined in \cref{sec:preliminaries}, with the
additional restriction that each good $g$ is positively valued by exactly
two agents, denoted by $i_g$ and $j_g$.\footnote{The model of
\citet{christodoulou2023fair} allows $v_{i_g}(g)=0$ or $v_{j_g}(g)=0$.
Requiring both endpoint valuations to be positive strengthens the negative
results established in this section. Moreover, the NP-hardness reduction
of \citet{christodoulou2023fair} for deciding whether an EFX orientation
exists already uses instances satisfying this restriction.}
Such an instance can be represented by an undirected graph $H=(N,M)$,
whose vertices correspond to agents and whose edges correspond to goods:
each good $g$ is represented by the edge $(i_g,j_g)$ joining the two agents
who value it positively. We therefore use the terms goods and edges
interchangeably. We define an \emph{orientation} as follows.

\begin{definition}\label{def:orientation}
Given an instance of the fair division problem on a graph $H=(N,M)$,
an \emph{orientation} is a full allocation $\X$ such that
$g \in X_{i_g} \cup X_{j_g}$ for every good $g=(i_g,j_g)\in M$.
\end{definition}

\noindent
In an orientation of a simple graph, each agent receives at most $n-1$
edges. Thus, it suffices to consider $k\in[1,n-1]$.
We say that a vertex $i$ \emph{receives an edge} $g$ in an allocation $\X$
if $g\in X_i$. Equivalently, an orientation $\X$ defines a directed graph
$H_d(\X)$ obtained by directing each edge of $H$ towards the vertex that
receives it.
\medskip

\noindent
Our first result gives a family of graph instances that admit no
\efkx orientation, even when $k=\Omega(n)$. The proof appears below.

\begin{theorem}\label{thm: non-existence-orientation}
Let $k=\Omega(n)$. There exist instances of the fair allocation problem
on graphs that admit no \efkx orientation.
\end{theorem}

\begin{proof}[Proof of \cref{thm: non-existence-orientation}]
    Consider an instance of the fair allocation problem on a graph $H=(N,M)$ with $n=4k+2$ nodes and the following two types of edges with symmetric valuations:
    \begin{quote}
        \emph{Heavy} edges with value $\alpha>4k+1$ and \emph{Light} edges with value $1$.
    \end{quote}
The set $M$ is defined as follows. Fix an ordering $1,\ldots,n$ of the nodes of the graph, and, for $i=1,\ldots,n-1$, add an edge $(i,i+1)$ according to this ordering. Additionally, add an edge $(n,1)$, essentially creating a (Hamiltonian) cycle $C$. For each edge $(i,i+1)$ in this ordering, let the edge be a \emph{heavy} edge if $i$ is even, and let the edge be a \emph{light} edge if $i$ is odd, and set the valuations of the endpoints of the edge accordingly (with both endpoints having the same value for the edge). For edge $(n,1)$, let the edge be a \emph{heavy} edge. Since $n$ is even, the edges of the cycle $C$ alternate between heavy and light edges. Finally, for any pair of nodes $\{i,j\}$ other than the pairs in $C$, add an edge $(i,j)$ to $M$ and let it be a \emph{light} edge, setting the values of the corresponding endpoint agents accordingly. In the end, the graph is the $K_{4k+2}$ graph.
\medskip

\noindent Assume by contradiction that an \efkx orientation $\X$ exists on $H$ and consider all the nodes and edges of $C$. In $\X$, exactly $2k+1$ nodes must receive the sole heavy edge that they are incident to; let $N_h$ denote the set of those nodes. Consider the $K_{2k+1}$ subgraph $H'=(N_h,M_h)$, where $M_h$ is the set of light edges between the nodes in $N_h$; see \cref{fig: ef2x-orientation-imposs} for an illustration in the case of $k=2$. Since $H'$ is a complete graph on $N_h$ nodes, it has a total of $(|N_h|\cdot(|N_h|-1))/2$ edges.
\medskip

\noindent We now consider the orientation of the edges in $M_h$ in $\X$, and let $H_d'(\X)$ be the corresponding directed graph of $\X$. Let $d_i^{-}$ be the \emph{in-degree} of any node $i\in N_h$ in this graph. By the total number of edges of $M_h$ calculated above, it follows that there must exist some node $i\in N$ with $d_i^{-}\geq\left\lfloor\frac{|N_h|-1}{2}\right\rfloor$. Since $|N_h|=2k+1$ by construction, this implies that $d_i^{-}\geq k$.
\medskip

\noindent Since $i\in N_h$, by definition the node received some heavy edge $(i,j)$, for some $j\in N$. Node $j$ does not receive its unique heavy edge $(i,j)$ and is incident to $n-2$ light edges. Therefore, we can bound $v_j(X_j)$ as:
\[
    v_j(X_j)\leq n-2=4k.
\]

\noindent Let $Y\subset X_i$ be a bundle consisting of $k$ light edges. Since $X_i$ also contains a heavy edge, it holds that $v_j(X_i\setminus Y)\geq\alpha>4k+1$. This implies that $\X$ is not \efkx, and we obtain a contradiction.
\end{proof}

\begin{remark}\label{rem:orientation-non-existence}
A closer look at the proof of \cref{thm: non-existence-orientation} reveals that the proof in fact establishes something stronger, namely that $\alpha$-\efkx orientations do not always exist, for any $\alpha \in (0,1]$. We elected to make the simpler statement about (exact) \efkx orientations instead, to be consistent with the NP-completeness result of \cref{thm:NP-c-orientation} below.
\end{remark}

\begin{figure}[H]
    \centering
    \includegraphics[width=0.5\linewidth]{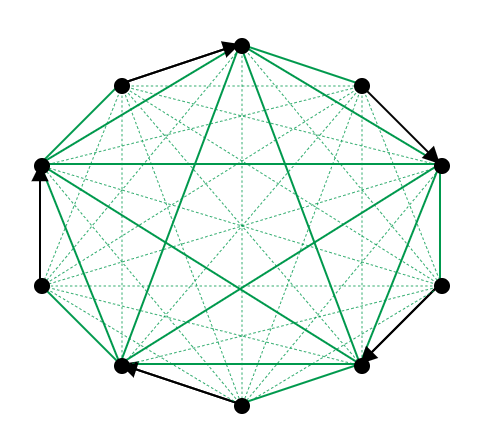}
    \caption{The graph $H=(N,M)$ used in the proof of \cref{thm: non-existence-orientation}, when $n=10$. The black edges denote the heavy edges and the green edges denote the light edges. $N_H$ is the set of nodes incident to the inner $K_5$ graph of light edges. We can easily verify that for any orientation of the highlighted inner $K_5$ graph, there exists a node in $N_h$ with at least two incoming edges.}
    \label{fig: ef2x-orientation-imposs}
\end{figure}

\noindent
Next, we establish the NP-completeness of deciding whether a graph
instance admits an \efkx orientation. Our NP-hardness proof proceeds by
induction on $k$, using the NP-completeness result of
\citet{christodoulou2023fair} for EFX orientations as the base case.
Membership in NP follows from polynomial-time verification of a given
orientation. We formally define the decision problem below and state the
corresponding theorem. The proof appears below.

\begin{center}
\fbox{%
\begin{minipage}{0.9\linewidth}
\underline{Computational Problem:} \textsc{\efkx Orientation}

\smallskip

\textbf{Input:} A graph $H=(N,M)$ and a value $v_i(g)$ for each vertex
$i\in N$ and each edge $g\in M$ incident to $i$.

\smallskip

\textbf{Output:} \textsc{Yes} if the instance admits an \efkx orientation,
and \textsc{No} otherwise.
\end{minipage}%
}
\end{center}

\begin{theorem}\label{thm:NP-c-orientation}
For every fixed integer $k\geq 1$,
\textsc{\efkx Orientation} is NP-complete.
\end{theorem}

\begin{proof}[Proof of \cref{thm:NP-c-orientation}]
    \noindent The membership of the problem in NP is immediate; indeed, if we are given an allocation $\X$, we can in linear time in $m$ check if it is a valid orientation. We next argue about the NP-hardness of the problem, via induction on $k$. Specifically, we will assume that the problem \textsc{EF(k-1)X Orientation} is NP-hard, and we will argue that \textsc{\efkx Orientation} is NP-hard as well. For the base case, it is already known from \citep{christodoulou2023fair} that \textsc{EF1X Orientation} (or, simply, \textsc{EFX Orientation}) is NP-hard. Thus, effectively, it suffices to construct a polynomial-time reduction from \textsc{EF(k-1)X Orientation} to \textsc{\efkx Orientation}. \medskip

    \noindent Let $I=\left(\hat{H},(\hat{v}_i)_{i \in \hat{N}}\right)$ be an instance of \textsc{EF(k-1)X Orientation} (with $\hat{H}=(\hat{N},\hat{M})$) and let $f(I) = \left(H,(v_i)_{i \in N}\right)$ be the corresponding instance of \textsc{\efkx Orientation} (with $H=(N,M)$) constructed via our reduction. Write $\hat{n}=|\hat{N}|$. The constructed instance contains $k$ disjoint copies of $I$. We describe the function $f$ below: \medskip
    
    \noindent First, we introduce an \emph{envy enhancer} gadget $H_e = (N_e, M_e)$ that consists of $|N_e| = 4(k-1) + 2$ nodes and two types of edges, \emph{heavy edges} and \emph{light edges}. The values of the edges for any agent incident to those are
    \begin{itemize}
    \item[-] $|N_e|+ \beta + 1$ for heavy edges, and
    \item[-] $1$, for light edges.
    \end{itemize}
    We choose $\beta=|N_e|+2$, so that $\beta > |N_e| + 1$ holds. \medskip
    
    \noindent Similarly to the proof of \cref{thm: non-existence-orientation}, we fix an ordering $1,\ldots,4(k-1)+2$ of the nodes of $H_e$, and, for $i=1,\ldots,4(k-1)+1$, we add an edge $(i,i+1)$ according to this ordering. Additionally, we add an edge $(4(k-1)+2,1)$, essentially creating a (Hamiltonian) cycle $C$. For each edge $(i,i+1)$ in this ordering, let the edge be a \emph{heavy} edge if $i$ is even, and let the edge be a \emph{light} edge if $i$ is odd, and set the values of the endpoints of the edge accordingly (with both endpoints having the same value for the edge). For edge $(4(k-1)+2,1)$, let the edge be a \emph{heavy} edge. Since $|N_e|$ is even, the edges of the cycle $C$ alternate between heavy and light edges. Finally, for any pair of nodes $\{i,j\}$ other than the pairs in $C$, we add an edge $(i,j)$ to $M_e$ and we let it be a \emph{light} edge, setting the values of the corresponding endpoint agents accordingly. \medskip

    \noindent Next, we introduce a \textit{funnel} gadget $F$ with $k+1$ central nodes $N_F=\{s,v_1,\ldots,v_k\}$ and two types of edges, \emph{solid} and \emph{transit} edges. The values of the edges for any agent incident to those are
    \begin{itemize}
    \item[-] $k \cdot \beta + 1$ for solid edges, and
    \item[-] $\beta$, for transit edges.
    \end{itemize}
    The \textit{funnel} gadget consists of $k$ vertices $v_1, v_2, \ldots, v_k$ which are connected to a vertex $s$ with a solid edge. For each $i\in[k]$, take a separate copy of the \textit{envy enhancer} gadget described above and connect $v_i$ to the $2k$ endpoints of $k$ heavy edges in that copy via transit edges. Let $M_F$ denote all solid and transit edges. When discussing a single enhancer copy, we use $N_e,M_e,N_K$ for its vertices, internal edges, and selected transit endpoints, respectively, without adding copy indices. \medskip

    \noindent Lastly, in each of the $k$ input copies, we add a \emph{connecting edge} $(s,j)$ for every node $j$ whose degree in $\hat{H}$ is strictly more than $k-1$. Let $M_c$ denote the set of connecting edges; thus $|M_c|\leq k\hat{n}$. The value of any connecting edge for either endpoint is $\delta>0$. Let $D$ be the product of the denominators of all rational\footnote{We assume that all endpoint values are rational numbers encoded in binary.} endpoint values $\hat{v}_i(g)$. We choose
    \begin{equation*}
        \delta=\frac{1}{D(k\hat{n}+1)}.
    \end{equation*}
    Since each denominator $b$ divides $D$, every input value
    $a/b=(aD/b)/D$ has an integer numerator over $D$, so additivity of the valuations allows any input edge value and bundle value to be written as $p/D$ and $q/D$, respectively, with integers $p,q$, and any strictly positive difference between them is
    $(p-q)/D\geq 1/D$ because $p-q$ is a positive integer. Moreover, $\delta<1/D$ and $|M_c|\delta<1$. The number $\delta$ has polynomial bit length and can be computed in polynomial time. See \cref{fig: efkx-hardness} for an illustration of the gadget. \medskip
     
    \begin{figure}[H]
    \centering
    \includegraphics[width=0.9\linewidth]{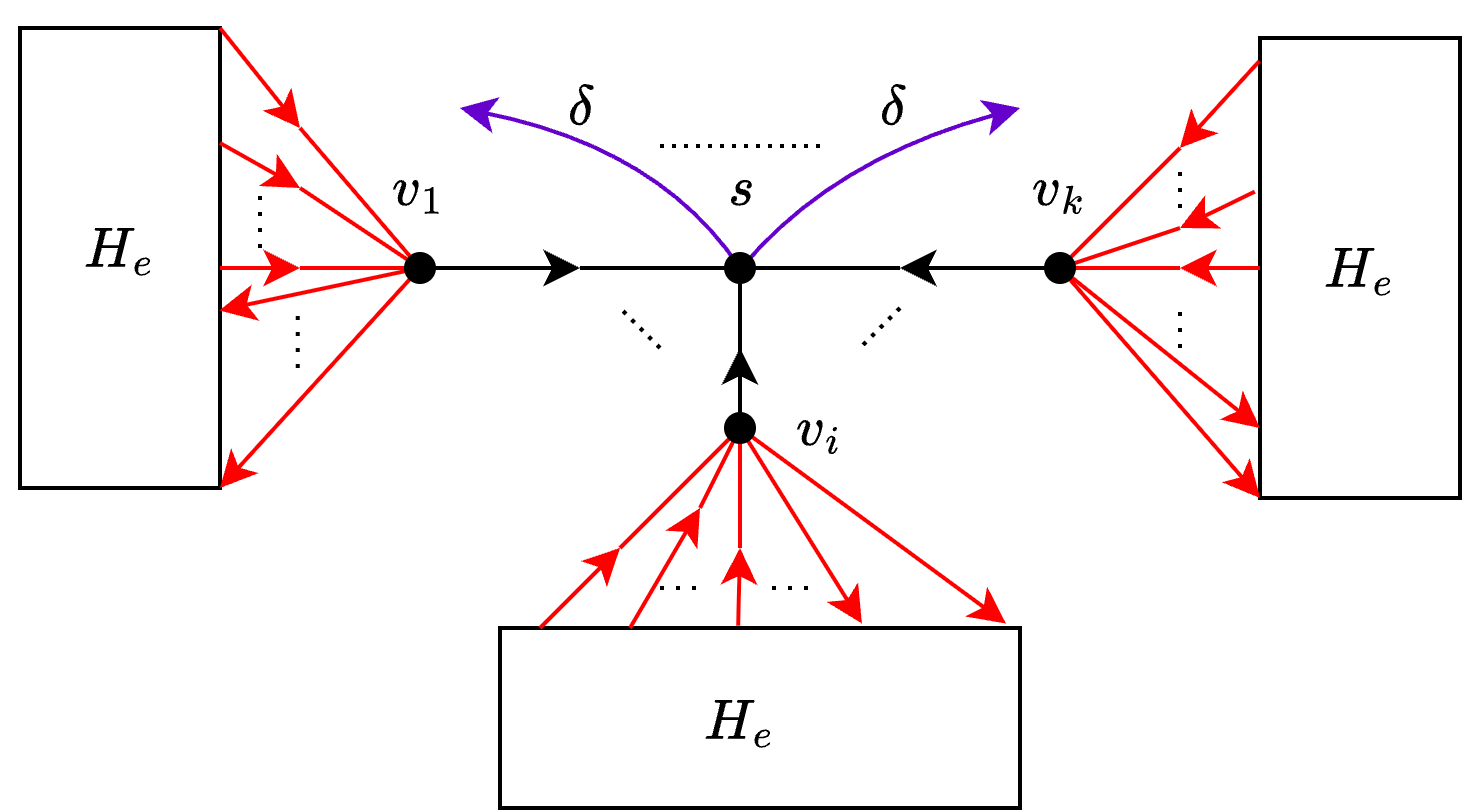}
    \caption{The \textit{funnel} gadget $F$ used in the proof of \cref{thm:NP-c-orientation}. Each $H_e$ box represents a separate \emph{envy enhancer} gadget. The black edges are \emph{solid} edges and the red edges are \emph{transit} edges connecting each $v_i$ to the endpoints $N_K$ of $k$ heavy matching pairs in its enhancer. The displayed orientation is the one used to extend an EF(k-1)X orientation of the input: the solid edges point towards $s$, $k$ transit edges point from $N_K^h$ towards $v_i$, and the other $k$ point towards $N_K^l$. The blue \emph{connecting} edges join $s$ to the vertices of degree at least $k$ in each of the $k$ input copies and point towards those vertices in this orientation. In an arbitrary \efkx orientation, at most $k-1$ connecting edges point towards $s$.}
    \label{fig: efkx-hardness}
    \end{figure}

    \noindent First, assume that an \efkx orientation $\X$ of $M$ exists. We will show that the restriction to one of the input copies is an EF(k-1)X orientation $\hat{\X}$ of $\hat{M}$. First, we establish the properties of the funnel and enhancer orientations that we need. \medskip
    
    \noindent Namely, in each envy enhancer gadget, using a similar argument to the proof of \cref{thm: non-existence-orientation}, we can establish that the $2k-1$ nodes receiving their incident heavy edge each receive exactly $k-1$ light edges. To see this, consider the set $N_h$ of the nodes that receive their incident heavy edge. The subgraph $H' = (N_h, M_h)$ induced has a total of $(|N_h| \cdot (|N_h| - 1))/2 = |N_h| \cdot (k - 1)$ edges. Therefore, we just need to prove that no node $i\in N_h$ can receive more than $k-1$ edges besides its heavy edge $g_i=(i,j)$. Node $j$ does not receive its heavy edge and is incident to $|N_e|-2$ light edges and at most one transit edge. Thus,
    \[
        v_j(X_j) \leq |N_e|-2+\beta.
    \]
    \noindent On the other hand, if node $i$ received at least $k$ edges besides $g_i$, we could choose $Y\subseteq X_i\setminus\{g_i\}$ with $|Y|=k$, giving
    \[
        v_j(X_i \setminus Y)
        \geq |N_e|+\beta+1
        >v_j(X_j),
    \]
    contradicting \efkx. The edges in $M_h$ therefore exhaust the total capacity of the nodes in $N_h$ for edges besides their heavy edges. Consequently, every node in $N_h$ receives exactly $k-1$ light edges from $M_h$ and no edge from outside $N_h$. \medskip
    
    \noindent Since $N_K$ consists of $k$ complete heavy matching pairs, exactly $k$ of its nodes receive their heavy edge. Write $N_K^h=N_K\cap N_h$ for these nodes. None receives its transit edge, so $v_i$ receives the corresponding $k$ transit edges. \medskip

    \noindent Next, consider the other $k$ nodes $N_K^l=N_K\setminus N_K^h$, which do not receive their incident heavy edge. If the transit edge $(j,v_i)$ of some $j\in N_K^l$ were assigned to $v_i$, then $j$ would receive neither its heavy edge nor its transit edge, so
    \[
        v_j(X_j) \leq |N_e|-2.
    \]
    \noindent Since $v_i$ already receives $k$ transit edges from $N_K^h$, it would then receive at least $k+1$ edges. Choosing $Y\subseteq X_{v_i}\setminus\{(j,v_i)\}$ with $|Y|=k$ would give
    \[
        v_j(X_{v_i} \setminus Y)
        \geq \beta
        >|N_e|-2
        \geq v_j(X_j),
    \]
    contradicting \efkx. Hence every node in $N_K^l$ receives its transit edge. \medskip

    \noindent Hence every vertex $v_i$ receives exactly $k$ transit edges, although it may also receive its solid edge. Lastly, we show that $s$ receives at most $k-1$ connecting edges. First, $s$ must receive at least one solid edge. Otherwise, each $v_i$ would receive its solid edge and its $k$ transit edges, while $v_s(X_s)\leq |M_c|\delta<1<k\beta+1$. Removing the $k$ transit edges from $X_{v_i}$ would violate \efkx for $s$. Now choose $v_i$ whose solid edge is assigned to $s$. This vertex receives only its $k$ transit edges, so
    \[
        v_{v_i}(X_{v_i})=k\beta.
    \]
    \noindent If $|X_s|>k$, we could choose $Y\subseteq X_s\setminus\{(v_i,s)\}$ with $|Y|=k$, obtaining
    \[
        v_{v_i}(X_s\setminus Y)
        \geq k\beta+1
        >v_{v_i}(X_{v_i}),
    \]
    contradicting \efkx. Therefore $|X_s|\leq k$. Since $s$ receives at least one solid edge, at most $k-1$ connecting edges are assigned to $s$. \medskip
     
    \noindent There are $k$ input copies and at most $k-1$ connecting edges assigned to $s$. Thus at least one copy has every connecting edge assigned to its endpoint. Fix such a copy, identify it with $\hat{H}$, and let $\hat{X}_i=X_i\cap\hat{M}$ for each of its nodes. Now consider a node $i\in\hat{N}$. If it is not incident to a connecting edge or receives at most $k-1$ edges, then $|\hat{X}_i|\leq k-1$, and no node can EF(k-1)X-envy it. Also, if an edge $\hat{e}=(i,j)$ is assigned to $j$, then $j$ values no edge in $\hat{X}_i$ positively. Therefore, it remains to consider $|\hat{X}_i|\geq k$ and an edge $\hat{e}=(i,j)$ assigned to $i$. \medskip
    
    \noindent We will show that $\hat{v}_j(\hat{X}_j)\geq\hat{v}_j(\hat{e})$. Assume that $\hat{v}_j(\hat{X}_j)<\hat{v}_j(\hat{e})$. Since $i$ receives at least $k$ edges and its connecting edge, $|X_i|\geq k+1$. Choose $Y\subseteq X_i\setminus\{\hat{e}\}$ with $|Y|=k$. The \efkx property gives $v_j(X_j)\geq v_j(X_i\setminus Y)\geq\hat{v}_j(\hat{e})$. On the other hand, $j$ receives at most one connecting edge, so
    \begin{align*}
        v_j(X_j)
        \leq \hat{v}_j(\hat{X}_j)+\delta 
        < \hat{v}_j(\hat{X}_j)+\frac{1}{D}
        \leq \hat{v}_j(\hat{e}).
    \end{align*}
    \noindent The last inequality follows because
    $\hat{v}_j(\hat{e})>\hat{v}_j(\hat{X}_j)$ and both values
    are integer multiples of $1/D$, implying that
    $\hat{v}_j(\hat{e})-\hat{v}_j(\hat{X}_j)\geq 1/D$. This contradicts $v_j(X_j)\geq\hat{v}_j(\hat{e})$. Therefore $\hat{v}_j(\hat{X}_j)\geq\hat{v}_j(\hat{e})$, which proves that $\hat{\X}$ is an EF(k-1)X orientation of $\hat{M}$. \medskip

    \noindent On the other hand, assume that an EF(k-1)X orientation $\hat{\X}$ of $\hat{M}$ exists. Use this orientation in each of the $k$ input copies. We will extend it to an \efkx orientation $\X$ of $M$. First, we orient all connecting edges in $M_c$ into their endpoints in the input copies. Next, we orient the edges $M_e$ in each envy enhancer as follows. For every edge $(i,i+1)$ for $i=1,\ldots,4(k-1)+1$ in the cycle $C$, orient it towards $i+1$, and orient the closing edge $(4(k-1)+2,1)$ towards $1$. Let $N_h$ denote the set of nodes that receive their heavy edge. Orient the light edges within $N_h$ as a regular tournament, so that each node receives exactly $k-1$ of them. Such an orientation exists since $|N_h|=2k-1$ is odd: label these nodes cyclically and direct each node towards its next $k-1$ nodes. Orient all light edges between $N_h$ and $N_e\setminus N_h$ towards the latter set, and orient light edges within the latter set arbitrarily. For the transit edges incident to the selected $k$ matching pairs, orient the $k$ edges from $N_K^h$ towards $v_i$ and the other $k$ edges towards $N_K^l$. Finally, orient every solid edge $(v_i,s)$ towards $s$. \medskip

    \noindent Every node in $N_h$, every $v_i$, and $s$ now receives exactly $k$ edges. Thus no comparison is required for a heavy or solid edge, or for a transit edge assigned to $v_i$. Each node in $N_e\setminus N_h$ receives at least $2k-2$ light edges from $N_h$, so every endpoint of a light edge has utility at least $1$. Every transit edge assigned to an enhancer node has its other endpoint $v_i$ owning value $k\beta\geq\beta$. Every connecting edge has its other endpoint $s$ owning value $k(k\beta+1)>\delta$. Hence all new edges satisfy the \efkx comparisons. For an edge $\hat{e}=(i,j)$ assigned to $i$ with $|X_i|>k$, we have $|\hat{X}_i|\geq k$, since $i$ receives at most one new edge. The EF(k-1)X property therefore gives
    \[
        v_j(X_j)
        \geq\hat{v}_j(\hat{X}_j)
        \geq\hat{v}_j(\hat{e})
        =v_j(\hat{e}).
    \]
    Therefore $\X$ is \efkx. For every fixed $k$, the graph has polynomial size and all values have polynomial bit length, so the reduction is computable in polynomial time. This completes the induction and the proof.
\end{proof}

\section{Future Work}
A natural follow-up question to our work is whether we can improve the approximation factor; for example, concretely, can we beat the $3/4$ approximation barrier established in our work for EF2X? What kind of improvement in the approximation factors can we achieve if we also add further restrictions on the valuation functions (e.g. three-valued instances)? Finally, now that our work settled the question of \efkx orientations, an interesting direction is settling the complexity of EF1 orientations, see \citep{deligkas2024ef1}.

\begin{ack}
Aris Filos-Ratsikas acknowledges support from the Engineering
and Physical Sciences Research Council (EPSRC) through the
project ``Algorithms and Complexity for Economic Environments''
(grant EP/Y003624/1, now EP/Y003624/2).
\end{ack}

\bibliographystyle{plainnat}
\bibliography{bio}








\clearpage 

\appendix
\raggedbottom

\onecolumn

\section*{\centering \Huge Appendix}
\vspace{1cm}

\section{Pseudocode for the algorithms and subroutines of \cref{sec:preliminaries}}
\crefalias{section}{appendixsection}
\label{app:section2-subroutines}

In this subsection, we provide the pseudocode for the Envy Cycle Elimination algorithm of \citep{lipton2004approximately}, as well as the subroutines defined in \cref{sec:preliminaries}. The former is identical to the one presented in \citep[Section 4.1]{amanatidis2024pushing}, and the latter are virtually identical to those presented in \citep[Section 2.3]{amanatidis2024pushing}. \medskip 

\SetAlgorithmName{ALGORITHM}{Algorithm}{List of Algorithms}

\begin{algorithm}[H]
	\DontPrintSemicolon 
    \caption{$\ECE(\X, G)$}\label{alg:ece}
        \KwData{A partial allocation $\X$ and its envy graph ${G}(\X)$.}
        \KwResult{A complete allocation $\X'$.}
        $M'\gets\PP(\X)$\;
		\For{every $g \in M'$ in arbitrary order}{
			\While{there is no source in $G({\X})$ }{
				Find a cycle $C$ in $G({\X})$ \;
                $\X \gets \CR(\X,{G}(\X),C)$
				}
			
			Let $s\in N$ be a source in $G({\X})$ \label{line:source} \; 
            $X_{s} \gets X_{s}\cup \{g\}$\;
		}
		\Return $\X$ \; 
\end{algorithm} 
 \medskip 

\SetAlgorithmName{SUBROUTINE}{Subroutine}{List of Subroutines}

\begin{algorithm}[H]
\DontPrintSemicolon
\caption{$\CR(\X,\tilde{G},C)$} \label{alg:cycle-resolution}
\SetKwComment{Comment}{/* }{ */}
\KwData{A partial allocation $\X$, its graph $\tilde{G}(\X)$, and a cycle $C$ in $\tilde{G}(\X)$}
\KwResult{An updated partial allocation $\X$ obtained by giving each
agent on $C$ the bundle held by her successor.\vspace{3pt}}
   $\tilde{\X}\gets \X$ \tcp*[r]{$\tilde{\X}$ is an auxiliary allocation}
   \For{every edge $(i,j) \in C$}
   {
   $X_i \gets \tilde{X}_j$ \tcp*[r]{swap the bundles backwards along the cycle}
   }
   \Return $\X$
\end{algorithm}
 \medskip 

\begin{algorithm}[H]
\DontPrintSemicolon
\caption{$\ACR(\X,\tilde{G})$} \label{alg:all-cycles-resolution}
\SetKwComment{Comment}{/* }{ */}
\KwData{A partial allocation $\X$ and its graph $\tilde{G}(\X)$.}
\KwResult{An updated partial allocation $\X$ such that $\tilde{G}(\X)$ is acyclic.\vspace{3pt}}
   \While{there exists a cycle $C$ in $\tilde{G}(\X)$}
   {
   $\X\gets\CR(\X,\tilde{G}(\X),C)$\;
   }
   \Return $\X$
\end{algorithm}
 \medskip 

\begin{algorithm}[H]
\DontPrintSemicolon
\caption{$\PR(\X,\tilde{G},\Pi)$} \label{alg:path-resolution}
\SetKwComment{Comment}{/* }{ */}
\KwData{A partial allocation $\X$, its graph $\tilde{G}(\X)$, and a path $\Pi = (i_1, i_2, \ldots, i_\ell)$ in $\tilde{G}(\X)$}
\KwResult{An updated set of bundles $X_i$, one for every agent $i \in \{i_1, i_2, \ldots, i_{\ell-1}\}$. \vspace{3pt}} 
\For(\tcp*[f]{go through every $i$ such that $(i,j) \in \Pi$ following the path}){$h \gets 1$\  to\  $\ell-1$} 
{
$X_{i_h} \gets X_{i_{h+1}}$ \tcp*[r]{assign to $i_h$ the bundle of the agent $i_{h+1}$ that she envies}
}
\Return $(X_i)_{i \in N: \exists (i, \ell) \in \Pi}$
\end{algorithm}

\section{The running time of the \GPPA algorithm}
\crefalias{section}{appendixsection}
\label{app:alternative-running-time-proof}

In this section, we provide some useful discussion related to the running time of the \GPPA algorithm. To this end, we provide the following remark regarding the proof of its running time, as presented in \citep{amanatidis2024pushing}.

\begin{remark}\label{rem:amanatidis-et-al-proof}
The proof of polynomial running time in \citep{amanatidis2024pushing}, as stated, is technically incorrect. This is due to an oversight in the argument about the termination of the \ACR subroutine when applied to the ``enhanced graph'' that the authors use in their version of their algorithm. In particular, \citet{amanatidis2024pushing} state that the \ECE procedure, when applied to the enhanced graph, will always result in a reduction in the number of edges. This was proved to not be the case by \citet{hv2025almost} via an appropriate counter-example. \citet{hv2025almost} also provide a fix, which however only results in a pseudopolynomial-time algorithm, rather than a polynomial-time one. In our modified envy graph $\Gr$, the statement about the reduction in the number of edges is true, and an almost straightforward adaptation of the argument in \citep{amanatidis2024pushing} goes through. Intuitively, this is because, the new edges that can appear in the enhanced graph in the counter-example of \citet{hv2025almost} are already present in the modified graph $\Gr$. Indeed, this oversight in the analysis in \citep{amanatidis2024pushing} is not a feature of the algorithm, which still runs in polynomial time.
\end{remark}

\section{Proof of Lemma~\ref{lem:GPA-satisfies-properties}}
\crefalias{section}{appendixsection}
\label{app:proofG3PA}

\begin{proof}
We will develop the proof in three steps: (1) we will first prove that $\X^1$ satisfies \cref{prop:a,prop:b,prop:c}, then (2) we will prove that $X^1$ satisfies \cref{prop:d,prop:e,prop:f}, and finally (3), we will prove the property about the sources of $\Ge(\X^1)$. The only part of the algorithm where \cref{prop:a} might be violated is if \GPPA returns a full allocation in Substep 6.2, in which case it will follow from our arguments that $\X^1$ is $\frac{k+1}{k+2}$-\efkx. 

\paragraph{(1) Satisfying \cref{prop:a,prop:b,prop:c}} We will prove that the algorithm satisfies \cref{prop:a,prop:b,prop:c} by induction on the number of executions of the main while loop of the algorithm. In particular, we will assume that the conditions hold before the $\ell$-th execution, and we will show that they still hold after the $\ell$-th execution (or, equivalently, before the $(\ell+1)$-th execution). The base case follows directly from the fact that $\X$ is a seed allocation, and hence by definition satisfies \cref{prop:a,prop:b,prop:c}. For ease of reference, we will use $\X$ and $\X'$ to denote the partial allocation before and after the execution of the while loop, respectively. Additionally, we will not explicitly argue about pairs of agents whose bundles remained unchanged during the current execution, even if those bundles exchanged owners in some graph reallocation step; \cref{prop:a,prop:b,prop:c} are still satisfied for such pairs, since they were satisfied before, by the induction hypothesis. Note that during the execution of the while loop, if step $i$ is executed, any step $j \neq i$ is not.  \medskip

\noindent \emph{If Step~\ref{step1} is executed:} \cref{prop:a} is satisfied for $\X$ by the induction hypothesis, and the only difference between $\X$ and $\X'$ is in the bundle of agent $i$, which however still contains only a single good. Therefore \cref{prop:a} is satisfied after the execution of the step. Since $v_i(X_i') > v_i(X_i)$ and the remaining bundles remain unaffected, it follows that \cref{prop:b,prop:c} are satisfied by the induction hypothesis. Furthermore, since $|X_i'|=1$, every agent $j \in N \setminus \{i\}$ is \efkx towards agent $i$, and thus \cref{prop:b,prop:c} are satisfied for that agent as well. Hence, \cref{prop:a,prop:b,prop:c} are satisfied for all agents. \medskip

\noindent \emph{If Step~\ref{step2} is executed:} For \cref{prop:a}, notice that all bundles in $\X$ remain unchanged, except the bundle of agent $i$, which changes in size from $k+1$ to $1$; hence $\X'$ satisfies \cref{prop:a}. By the induction hypothesis, agent $i$ was $\frac{k+1}{k+2}$-\efkx towards any agent $j$ in $\X$. Since $v_i(X_i') > \frac{k+2}{k+1}v_i(X_i)$, agent $i$ is \efkx towards any other agent in $\X'$, hence satisfying \cref{prop:b,prop:c}. Since $|X_i'|=1$, every agent $j \in N \setminus \{i\}$ is \efkx towards agent $i$, and thus \cref{prop:b,prop:c} are satisfied for that agent as well. Hence, \cref{prop:a,prop:b,prop:c} are satisfied for all agents. \medskip 

\noindent \emph{If Steps~\ref{step1} and \ref{step2} are \emph{not} executed:} We will proceed by considering the cases where one of Steps~\ref{step3} to \ref{step8} is executed. This means in particular that Steps~\ref{step1} and \ref{step2} were \emph{not} executed. From this, we can infer the following \emph{Step Properties}, for the allocation $\X$, just before the execution of some subsequent step $i$:

\begin{sproperties}
\item For any agent $j \in N$ with $|X_j|=1$, $v_j(X_j) \geq v_j(g)$ for all $g \in \PP(\X)$, as otherwise Step~\ref{step1} would have been executed. \label{prop:step-property-1}
\item For any agent $j \in N$ with $|X_j|=k+1$, $v_j(X_j) \geq \frac{k+1}{k+2}\cdot v_j(g)$ for all $g \in \PP(\X)$, as otherwise Step~\ref{step2} would have been executed. \label{prop:step-property-2}
\end{sproperties}

\noindent \emph{If Step~\ref{step3} is executed:} For \cref{prop:a}, notice that all bundles in $\X$ remain unchanged, except the bundle of agent $i$, which changes in size from $1$ to $k+1$; hence $\X'$ satisfies \cref{prop:a}. Consider agent $i$, and notice that $v_i(X_i') > \frac{k+1}{k+2}v_i(X_i)$, i.e., the value of the agent was possibly reduced in $\X'$, but no more than $(k+1)/(k+2)$ of her value in $\X$. By the induction hypothesis, given that $|X_i|=1$, the agent was \efkx towards any agent $j \in N$ in $\X$. By the inequality above, agent $i$ is now $\frac{k+1}{k+2}$-\efkx towards any such agent $j$; given that $|X_i'|=k+1$, \cref{prop:b,prop:c} are thus satisfied for the agent.  

Now consider any other agent $j \in N \setminus \{i\}$; we will consider two cases depending on the size of agent $j$'s bundle in $\X'$ (which the same as her bundle in $\X$):
\begin{itemize}[leftmargin=*]
    \item[-] \emph{Case 1: $|X_j|=1$.} In this case, we have that for all $\hat{g} \in X_i'$, $v_j(X_j) \geq v_j(\hat{g})$; this follows from \cref{prop:step-property-1} and because in $\X$, any of the goods of $X_i'$ were in $\PP(\X)$. This implies that agent $j$ is \efkx towards agent $i$.  
    \item[-] \emph{Case 2: $|X_j|=k+1$.} In this case, we have that for all $\hat{g} \in X_i'$, $v_j(\hat{g}) \leq \frac{k+2}{k+1} v_j(X_j)$; this follows from \cref{prop:step-property-2} and because in $\X$, any of the goods of $X_i'$ were in $\PP(\X)$. This implies that agent $j$ is $\frac{k+1}{k+2}$-\efkx towards agent $i$. 
\end{itemize}
In any case, \cref{prop:b,prop:c} are satisfied for agent $j$ also. Hence, all three properties are satisfied for all agents. \medskip

\noindent \emph{If Step~\ref{step4} is executed:}  For \cref{prop:a}, notice that the size of any bundle in $\X$ remains unchanged in $\X'$, and hence \cref{prop:a} is satisfied by the induction hypothesis. First, consider agent $i$ and observe that $v_i(X_i') > v_i(X_i)$. Since the size of agent $i$'s bundle did not change, \cref{prop:b,prop:c} still hold for the agent by the induction hypothesis. Now, consider any other agent $j \in N \setminus \{i\}$; we will consider two cases depending on the size of agent $j$'s bundle in $\X'$ (which the same as her bundle in $\X$): 
\begin{itemize}[leftmargin=*]
    \item[-] \emph{Case 1: $|X_j|=1$.} In this case, we have that for the good $g \in X_i'\setminus X_i$ that was added to agent $i$'s bundle, it holds that $v_j(X_j) \geq v_j(g)$;  this follows from \cref{prop:step-property-1} and because in $\X$, $g \in \PP(\X)$. Additionally, we know that $v_j(X_j) \geq v_j(\hat{g})$, for any good $\hat{g} \in X_i$; this follows from \cref{prop:b} and the induction hypothesis, as in $\X$, agent $j$ was \efkx towards agent $i$, whose bundle size was $k+1$. This implies that agent $j$ is \efkx towards agent $i$, and \cref{prop:b,prop:c} are satisfied for this agent. 
    \item[-] \emph{Case 2: $|X_j|=k+1$.} In this case, we have that for the good $g \in X_i'\setminus X_i$ that was added to agent $i$'s bundle, it holds that $v_j(X_j) \geq \frac{k+1}{k+2}v_j(g)$;  this follows from \cref{prop:step-property-2} and because in $\X$, $g \in \PP(\X)$. Additionally, we know that $v_j(X_j) \geq \frac{k+1}{k+2}v_j(\hat{g})$, for any good $\hat{g} \in X_i$; this follows from \cref{prop:c} and the induction hypothesis, as in $\X$, agent $j$ was $\frac{k+1}{k+2}$-\efkx towards agent $i$, whose bundle size was $k+1$. This implies that agent $j$ is $\frac{k+1}{k+2}$-\efkx towards agent $i$, and \cref{prop:b,prop:c} are satisfied for this agent. 
\end{itemize}
In any case, \cref{prop:b,prop:c} are satisfied for agent $j$ also. Hence, all three properties are satisfied for all agents. \medskip

\noindent \emph{If Step~\ref{step5} is executed:} First, observe that this step only reallocates bundles that were previously allocated, so \cref{prop:a} follows immediately from the induction hypothesis. Next, consider any agent $i \in N$ for whom $X_i' \neq X_i$; for any other agent, her bundle has not changed, and hence \cref{prop:b,prop:c} are satisfied by the induction hypothesis. Let $j$ be the agent $j \in N \setminus\{i\}$ such that $X_i'=X_j$. We consider two cases:
\begin{itemize}[leftmargin=*]
    \item[-] \emph{Case 1: $|X_i|=1$, $|X_i'|=k+1$.} By the definition of $\Gr$, it holds that $v_i(X_i') = v_i(X_j) > \frac{k+1}{k+2}\cdot v_i(X_i)$. By the induction hypothesis, agent $i$ was \efkx towards any agent $j \in N$, and hence the agent is now $\frac{k+1}{k+2}$-\efkx towards any such agent $j$, satisfying \cref{prop:b,prop:c}.
     \item[-] \emph{Case 2: $|X_i|=k+1$, $|X_i'|=1$.} By the definition of $\Gr$, it holds that $v_i(X_i') = v_i(X_j) > \frac{k+2}{k+1} \cdot v_i(X_i)$. By the induction hypothesis, agent $i$ was $\frac{k+1}{k+2}$-\efkx towards any agent $j \in N$, and hence the agent is now \efkx towards any such agent $j$, satisfying \cref{prop:b,prop:c}.
     \item[-] \emph{Case 3: $|X_i'|=|X_i|$.}
    The proxy utility strictly increases while the bundle size
    is unchanged, so $v_i(X_i')>v_i(X_i)$, and
    \cref{prop:b,prop:c} follow from the induction hypothesis.
\end{itemize}
In any case, \cref{prop:b,prop:c} are satisfied for agent $i$ also. Hence, all three properties are satisfied for all agents. \medskip

\noindent \emph{If Step~\ref{step6} is executed:} Assume first that Substep~6.1 is executed, which means that $\X$ is a partial allocation that needs to satisfy \cref{prop:a,prop:b,prop:c}. \cref{prop:a} is satisfied by the induction hypothesis, since the bundle of agent $s$ changed from having size $1$ to having size $k+1$, and all remaining bundles remained unchanged. We next argue for \cref{prop:b,prop:c}. 

Observe that for agent $s$, we have that $v_s(X'_s) \geq v_s(X_s)$. Since $|X_s|=1$, by the induction hypothesis, agent $s$ was \efkx towards any other agent in $\X$, and therefore the same holds for the agent in $\X'$; hence the properties are satisfied for this agent. Now consider any other agent $j \in N \setminus\{s\}$; we consider two cases depending on the size of agent $j$'s bundle in $\X'$ (which is the same as her bundle in $\X$):
\begin{itemize}[leftmargin=*]
 \item[-] \emph{Case 1: $|X_j|=1$.} By \cref{prop:step-property-1}, it holds that $v_j(X_j) \geq v_j(g)$ for any $g \in \PP(\X)$, and therefore $v_j(X_j) \geq v_j(g)$ for any good $g \in X_s'\setminus X_s$. Since $s$ is a source of $\Gr(\X)$, it holds that $v_j(X_j) \geq v_j(X_s)$, as otherwise there would be an edge $(j,s) \in \Gr(\X)$. Also, since $|X_s|=1$, from the discussion above, it follows that $v_j(X_j) \geq v_j(g)$ for any good $g \in X_s'$, and hence agent $j$ is \efkx towards agent $s$, thus satisfying properties \cref{prop:b,prop:c}.  
    \item[-] \emph{Case 2: $|X_j|=k+1$.} By \cref{prop:step-property-2}, it holds that $v_j(X_j) \geq \frac{k+1}{k+2}\cdot v_j(g)$ for any $g \in \PP(\X)$, and therefore $v_j(X_j) \geq \frac{k+1}{k+2}\cdot v_j(g)$ for any good $g \in X_s'\setminus X_s$. Since $s$ is a source of $\Gr(\X)$, it holds that $v_j(X_j) \geq \frac{k+1}{k+2}\cdot v_j(X_s)$, as otherwise there would be an edge $(j,s) \in \Gr(\X)$. Also, since $|X_s|=1$, from the discussion above, it follows that $v_j(X_j) \geq \frac{k+1}{k+2}\cdot v_j(g)$ for any good $g \in X_s'$, and hence agent $j$ is $\frac{k+1}{k+2}$-\efkx towards agent $s$, thus satisfying properties \cref{prop:b,prop:c}.  
\end{itemize}
In any case, \cref{prop:b,prop:c} are satisfied for agent $j$ also. Hence, all three properties are satisfied for all agents. Now assume that Substep 6.2 is executed instead. In that case $\X'=\X^1$ is a full allocation, and hence we need to establish that it is also $\frac{k+1}{k+2}$-\efkx. This follows from exactly the same arguments used above for agent $s \in N$ and any other agent $j \in N \setminus \{s\}$.   \medskip

\noindent \emph{If Step~\ref{step7} is executed:}
First, every agent on $\Pi$ other than $i$ receives her
successor's bundle, every agent outside $\Pi$ retains her
bundle, and agent $i$ receives $Y$.
By the induction hypothesis, all previously allocated bundles
have size either $1$ or $k+1$.
Agent $i$ receives a bundle of size $k+1$, and therefore
\cref{prop:a} is satisfied for all agents.
We will next argue for \cref{prop:b,prop:c}.

First, consider agent $i$ who receives the bundle $Y$ of goods; by definition of the step, we have that $v_i(X_i')=v_i(Y) > \frac{k+1}{k+2}\cdot v_i(X_i)$, as the \PRPA subroutine allocates bundle $Y$ to agent $i$. By the induction hypothesis, since $|X_i|=1$, agent $i$ was \efkx towards any agent $j \in N$. By the inequality above, the agent is $\frac{k+1}{k+2}$-\efkx towards any such agent $j$ in $\X'$; since $|X_i'|=k+1$, this establishes \cref{prop:b,prop:c} for the agent.  Now consider any other agent $j \in N \setminus\{i\}$; we consider two cases depending on the size of agent $j$'s bundle in $\X$, and for each of those, we consider two subcases, depending on the size of the agent's bundle in $\X'$: 
\begin{itemize}[leftmargin=*]
 \item[-] \emph{Case 1: $|X_j|=1$.} In this case, it follows from the induction hypothesis that in $\X$, agent $j$ was \efkx towards agent $s$, and hence it holds that $v_j(X_j) \geq v_j(g)$ for any good $g \in X_s$. Additionally, by \cref{prop:step-property-1}, we have that $v_j(X_j) \geq v_j(g)$ for any good $g \in \PP(\X)$. These two facts together establish that 
 \begin{equation}
     v_j(X_j) \geq v_j(g) \text{ for any good } g \in Y,
     \label{eq:step7-equation}
 \end{equation}
 since $Y$ has size $k+1$ and consists of goods from either $X_s$ or $\PP(\X)$.
 \begin{itemize}[leftmargin=*]
     \item[-] \emph{Subcase 1a: $|X_j'|=1$.} In this case, the edge between agent $j$ and the agent $\hat{j}$ whose bundle $j$ received by the \PRPA subroutine in $\Gr(\X)$ signifies (exact) envy, and hence we have that $v_j(X_j')=v_j(X_{\hat{j}}) \geq v_j(X_j)$. Combining this fact with \cref{eq:step7-equation}, we establish that $v_j(X_j')\geq v_j(g)$ for any good $g \in Y$. This implies that agent $j$ is \efkx towards agent $i$, and hence \cref{prop:b,prop:c} are satisfied for this agent.  
     \item[-] \emph{Subcase 1b: $|X_j'|=k+1$.} In this case, the edge between agent $j$ and the agent $\hat{j}$ whose bundle $j$ received by the \PRPA subroutine in $\Gr(\X)$ signifies $(k+1)/(k+2)$-approximate envy, and hence we have that $v_j(X_j')=v_j(X_{\hat{j}}) \geq \frac{k+1}{k+2} \cdot v_j(X_j)$. Combining this fact with \cref{eq:step7-equation}, we establish that $v_j(X_j')\geq \frac{k+1}{k+2}\cdot v_j(g)$ for any good $g \in Y$. This implies that agent $j$ is $\frac{k+1}{k+2}$-\efkx towards agent $i$, and hence \cref{prop:b,prop:c} are satisfied for this agent.
 \end{itemize}
 \item[-] \emph{Case 2: $|X_j|=k+1$.} In this case, it follows from the induction hypothesis that in $\X$, agent $j$ was \aefkx towards agent $s$, and hence it holds that $v_j(X_j) \geq \frac{k+1}{k+2}\cdot v_j(g)$, for any good $g \in X_s$. Additionally, by \cref{prop:step-property-2}, we have that $v_j(X_j) \geq \frac{k+1}{k+2}\cdot v_j(g)$ for any good $g \in \PP(\X)$. These two facts together establish that 
 \begin{equation}
     v_j(X_j) \geq \frac{k+1}{k+2}\cdot v_j(g) \text{ for any good } g \in Y,
     \label{eq:step7-equation-b}
 \end{equation} 
 since $Y$ has size $k+1$ and consists of goods from either $X_s$ or $\PP(\X)$.
  \begin{itemize}[leftmargin=*]
     \item[-] \emph{Subcase 2a: $|X_j'|=1$.} In this case, the edge between agent $j$ and the agent $\hat{j}$ whose bundle $j$ received by the \PRPA subroutine in $\Gr(\X)$ signifies $(k+2)/(k+1)$-approximate envy, and hence we have that $v_j(X_j')=v_j(X_{\hat{j}}) \geq \frac{k+2}{k+1}\cdot v_j(X_j)$. Combining this fact with \cref{eq:step7-equation-b}, we establish that $v_j(X_j')\geq v_j(g)$ for any good $g \in Y$. This implies that agent $j$ is \efkx towards agent $i$, and hence \cref{prop:b,prop:c} are satisfied for this agent.  
     \item[-] \emph{Subcase 2b: $|X_j'|=k+1$.} In this case, the edge between agent $j$ and the agent $\hat{j}$ whose bundle $j$ received by the \PRPA subroutine in $\Gr(\X)$ signifies (exact) envy, and hence we have that $v_j(X_j')=v_j(X_{\hat{j}}) \geq v_j(X_j)$. Combining this fact with \cref{eq:step7-equation-b}, we establish that $v_j(X_j')\geq \frac{k+1}{k+2}\cdot v_j(g)$ for any good $g \in Y$. This implies that agent $j$ is \aefkx towards agent $i$, and hence \cref{prop:b,prop:c} are satisfied for this agent.
 \end{itemize}
\end{itemize}
In summary, \cref{prop:a,prop:b,prop:c} are satisfied for all agents. \medskip 
 
\noindent \emph{If Step~\ref{step8} is executed:} In this step, the algorithm terminates. In particular, the step does not change the allocation $\X$ and \cref{prop:a,prop:b,prop:c} are thus satisfied by the induction hypothesis. 

\paragraph{(2) Satisfying \cref{prop:d,prop:e,prop:f}} Now that we have established that the \GPPA algorithm satisfies \cref{prop:a,prop:b,prop:c}, we will establish that it satisfies the remaining three properties, unless it returned a full allocation in Step~6.2. To this end, assume by contradiction that one of the properties is violated by the allocation $\X^1$ outputted by the algorithm after its termination; we consider three cases, depending on which property that is. \medskip

\noindent \emph{If \cref{prop:d} is violated.} In this case, there exists some agent $i \in N$ and some good $g \in \PP(\X^1)$ such that $v_i(X^1_i) < v_i(g)$. If $|X^1_i|=1$, then Step~\ref{step1} of the algorithm would have been executed, and the algorithm would not have terminated. If $|X^1_i|=k+1$, either Step~\ref{step2} or Step~\ref{step4} would have been executed, and the algorithm would not have terminated. In either case we obtain a contradiction. \medskip

\noindent \emph{If \cref{prop:e} is violated.} In this case, there exists some agent $i \in N$ with $|X^1_i|=k+1$ that has at least one $\frac{1}{k+1}$-critical good, i.e., there exists some good $g^{*} \in \PP(\X^1)$ such that $v_i(g^{*})>\frac{1}{k+1}\cdot v_i(X^1_i)$. Since $|X^1_i|=k+1$, this implies that $v_i(g^{*}) > \min_{g \in X^1_i}v_i(g)$, i.e., there is some good from the pool that the agent strictly prefers to one of the goods in her bundle. In such a case, however, either Step~\ref{step2} or Step~\ref{step4} of the algorithm would have been executed, contradicting its termination. \medskip 

\noindent \emph{If \cref{prop:f} is violated.} In this case, there is an agent $i \in N$ with $|X^1_i|=1$ and a set $Y_i^c$ of $\frac{1}{k+1}$-critical goods for agent $i$ such that
\begin{itemize}[leftmargin=*]
\item[-] $|Y_i^c| \geq k+1$, or 
\item[-] $|Y_i^c| \leq k$, but agent $i$ values the set $Y_i^c$ by more than a $\frac{k+1}{k+2}$-fraction of her value for her bundle $X_i^1$, i.e., $v_i(Y_i^c) > \frac{k+1}{k+2}\cdot v_i(X_i^1)$.    
\end{itemize}
In the former case, since $|X_i^1|=1$, Step~\ref{step3} of the algorithm would have been executed for agent $i$ and the goods $g_1,\ldots,g_{k+1} \in Y_i^c$, and the algorithm would not have terminated, a contradiction. In the latter case, the modified envy graph is acyclic,
since otherwise Step~\ref{step5} would apply, so there is
a path $\Pi=(s,\ldots,i)$ from a source $s$ to $i$.
Every source has $|X_s^1|=k+1$, since otherwise
Step~\ref{step6} would apply.
Choose $Z\subseteq X_s^1$ with
$|Z|=k+1-|Y_i^c|$ and set $Y=Y_i^c\cup Z$.
Then $|Y|=k+1$ and
$v_i(Y)\geq v_i(Y_i^c)>
\frac{k+1}{k+2}v_i(X_i^1)$, so Step~\ref{step7}
would apply, contradicting termination.\medskip

\noindent Therefore, any of \cref{prop:d,prop:e,prop:f} must be satisfied by $\X^1$, as otherwise the algorithm would not have terminated.

\paragraph{(3) The property of the sources of $\Gr$} We conclude the proof of \cref{lem:GPA-satisfies-properties} by arguing that, unless $\X^1$ is a full allocation, its modified graph $\Gr(\X^1)$ has at least one source, and for every source $s$, we have that $|X_s^1|=k+1$. The existence of a source follows directly from Step~\ref{step5} of the algorithm, which would have been executed if a cycle was present in $\Gr(\X^1)$. To establish the cardinality of the bundles of the sources, note that, by \cref{prop:a}, either $|X_s^1|=k+1$ or $|X_s^1|=1$, for any source $s$. In the latter case, however, Step 6.1. of the algorithm would have been executed, contradicting its termination. Therefore, it must be the case that $|X_s^1|=k+1$. 
\end{proof}

\section{Achieving $2/3$-EFX for 8 agents}
\crefalias{section}{appendixsection}
\label{app:8-agents}
In this section, we provide the proof of \cref{thm:8-agents-23-efx}, namely that there exists an algorithm that computes a $2/3$-EFX allocation for an instance with at most 8 agents in polynomial time. We refer to this algorithm as \IFAA, see Algorithm~\ref{alg:approx-efx}. The algorithm uses the same three general parts described in \cref{sec:efkx-approximations} for achieving approximations for any $k \geq 2$, but the second part is now crucially split into two sub-parts. Since in this section, the value of $k$ is fixed at $1$, we use the term \emph{critical good} instead of $1/2$-critical good for simplicity. The parts are the following.

\begin{itemize}
\item[1.] A polynomial-time algorithm \GPPAplus which produces a seed allocation $\X^0$ as input, and outputs a $2/3$-EFX allocation $\X^1$. We describe the algorithm more generally, but instantiate it with $k=1$ for the results of this section. We remark that \GPPAplus is a generalization of the \PPAplus algorithm of \citet{amanatidis2024pushing} which is used in the algorithm that achieves $2/3$-EFX allocations for $7$ agents in their work.\footnote{Our results can also be obtained by using the \PPAplus algorithm as presented in \citep{amanatidis2024pushing} verbatim as the first part of the \IFAA algorithm. We elect to present our version of the algorithm, since it is conceptually simpler and for consistency with the results in \cref{sec:efkx-approximations}, primarily due to the use of the modified graph $\Gr$ rather than the ``reduced envy graph'' and ``enhanced envy graph'' of \citep{amanatidis2024pushing}. } 
We present this algorithm in \cref{sec:G3PA+}, see Algorithm~\ref{alg:G3PA+}.
The output $\X^1$ also satisfies \cref{prop:a,prop:b,prop:c,prop:d,prop:e,prop:f}.

\item[2a.]A polynomial-time algorithm which inputs the partial allocation $\X^1$ and outputs a partial allocation $\X^2$. $\X^2$ is still $2/3$-EFX, and in addition, it does not contain any \emph{contested} critical good for any agent in $\PP(\X^2)$, i.e., any good that is critical for more than one agent.
We present this algorithm in \cref{sec:acc}.

\item[2b.] A polynomial-time algorithm which inputs the partial allocation $\X^2$ and outputs a partial allocation $\X^3$. $\X^3$ is an $2/3$-EFX partial allocation, where all the remaining (uncontested) critical goods in $\PP(\X^2)$ are allocated. We present this algorithm in \cref{sec:ucc}.

\item[3.] By the P2FA Lemma (\cref{lem:partial-efx-to-full-efx}), a polynomial-time algorithm which inputs the partial allocation $\X^3$ and outputs a full $2/3$-EFX allocation $\tilde{\X}$.
\end{itemize}

\subsection{The \GPPAplus Algorithm}\label{sec:G3PA+}
We first present the \GPPAplus algorithm. The algorithm is almost the same as the \GPPA algorithm, but augmented with an additional new step (Step~\ref{step8*}). \medskip

\noindent Whenever Step~\ref{step8*} is executed, \cref{prop:a,prop:b,prop:c,prop:d,prop:e,prop:f} still hold for the same reason they hold whenever Step~\ref{step7} is executed.
Hence, with minor modifications to the proof of \cref{lem:GPA-satisfies-properties}, we can have its analogue for \GPPAplus in \cref{lem:GPA+-satisfies-properties}.

\begin{lemma}\label{lem:GPA+-satisfies-properties}
    Let $\X^0$ be a seed allocation, and let $\X^1 = \GPPAplus((v_i)_{i \in N},\X)$ be the outcome of the \GPPAplus algorithm on input $\X$.
    Then either (i) $\X^1$ is a $\frac{k+1}{k+2}$-\efkx complete allocation or (ii) $\X^1$ satisfies \cref{prop:a,prop:b,prop:c,prop:d,prop:e,prop:f}, and, furthermore, the modified graph $\Ge(\X^1)$ has at least one source, and every source $s$ in $\Ge(\X^1)$ has $|X_s^1|=k+1$. 
\end{lemma}

\noindent Since we use the \GPPAplus algorithm with $k=1$, the properties of the output $\X^1$ can be restated as follows:
\begin{tcolorbox}[
    standard jigsaw,
    title=Desired Properties of a Partial Allocation $\X^1$:,
    opacityback=0,  
]
\begin{properties}[topsep=5pt,itemsep=0.5ex,partopsep=1ex,parsep=1ex]
\item For every agent $i \in N$, we have either than $|X^1_i|=1$ or $|X^1_i|=2$. \label{prop:a*}
\item Every agent $i \in N$ with $|X^1_i|=1$ is EFX towards any other agent.  
\label{prop:b*}
\item Every agent $i \in N$ is $2/3$-EFX towards any other agent. 
\label{prop:c*}
\item Every agent $i$ values her own bundle at least as much as any single unallocated good, i.e., for every agent $i \in N$ and good $g \in \PP(\X^1)$, $v_i(X^1_i)\geq v_i(g)$.\label{prop:d*}
\item Every agent $i \in N$ with $|X^1_i|=2$ does not have any critical goods, i.e., for any good $g \in \PP(\X^1)$, we have that $v_i(g)\leq \frac{1}{2} v_i(X^1_i)$.\label{prop:e*}
\item Any agent $i$ with $|X^1_i|=1$ has at most one critical good, and she values the good at most $\frac{2}{3}$ of her current bundle, i.e., there exists at most one good  $g \in \PP(\X^1)$ such that $\frac{1}{2}v_i(X^1_i)\leq v_i(g) \leq  \frac{2}{3} v_i(X^1_i)$.\label{prop:f*}
\end{properties}
\end{tcolorbox}

\SetAlgorithmName{ALGORITHM}{Algorithm}{List of Algorithms}
\begin{algorithm}[H]
\DontPrintSemicolon
\caption{\IFAA($(v_i(g))_{i \in N,\, g \in M}$)}
\label{alg:approx-efx}
\SetKwComment{Comment}{/* }{ */}
\KwData{Valuations $(v_i(g))_{i \in N,\, g \in M}$.}
\KwResult{A $2/3$-\efx allocation $\X$.\vspace{3pt}}

Let $\X^0$ be an allocation where each bundle has size $1$\;
\tcp{Start with an initial seed allocation (\cref{def:seed-allocation}) $X^0$.}
$\X^1 \gets \GPPAplus((v_i(g))_{i \in N,\, g \in M}, \X^0, 1)$ \;
\tcp{Run the \GPPAplus algorithm (Algorithm~\ref{alg:G3PA+}) on $\X^0$ to obtain a partial allocation $\X^1$ satisfying \cref{prop:a,prop:b,prop:c,prop:d,prop:e,prop:f} (\cref{sec:G3PA+}).}
\If{$\PP(\X^1)=\emptyset$}{
    \Return $\X^1$\;
}
\If{There is at least one contested critical good in $\PP(\X^1)$}{
$\X^2 \gets \ACC((v_i(g))_{i \in N,\, g \in M}, \X^1)$\;
\tcp{Run the \ACC algorithm (Algorithm~\ref{alg:acc}) on $\X^1$ to obtain a partial allocation $\X^2$ that is $2/3$-EFX and does not have any contested critical goods (\cref{sec:acc}).}}
\Else{$\X^2\gets \X^1$}
$\X^3 \gets \UCC(\X^2,G(\X^2))$\;
\tcp{Run the \UCC algorithm (Algorithm~\ref{alg:ucc}) to obtain a partial allocation $\X^3$ that is $2/3$-EFX and does not have any critical goods (\cref{sec:ucc}).}
$\X \gets \ECE(\X^3,G(\X^3))$\;
\tcp{Run the \ECE algorithm (Algorithm~\ref{alg:ece}) to obtain a full $2/3$-EFX allocation (\cref{lem:partial-efx-to-full-efx}).}
\Return $\X$
\end{algorithm}

\SetAlgorithmName{ALGORITHM}{Algorithm}{List of Algorithms}

\begin{algorithm}[H]
\DontPrintSemicolon
\caption{\GPPAplus$((v_i(g))_{i \in N,\, g \in M}, \X^0, k)$  }
\label{alg:G3PA+}
\SetKwComment{Comment}{/* }{ */}
\SetKw{Continue}{continue}
\SetKw{Break}{break}
\SetKw{Step}{Step}
\SetKw{Substep}{Substep}
\SetKwData{Kw}{}
\textbf{Input/Output:}(as in \GPPA)\;
... (same as \GPPA)...\;
\setcounter{AlgoLine}{7}
\nl \Kw{\texttt{Step 8}}\label{step8*}
\uElseIf{for some agent $i$ with $|X_i|=k+1$, there exists
a path $\Pi=(s,\ldots,i)$ in $\Gr(\X)$ starting at a
source $s$ of $\Gr(\X)$ and a set
$Y\subseteq X_s\cup\PP(\X)$ with $|Y|=k+1$ such that
$v_i(Y)>v_i(X_i)$}{
    $\X\gets\PRPA(\X,\Gr(\X),\Pi,Y)$\;
\tcp{Else if there exists a path from a source $s$ to some agent $i$ such that $|X_i|=k+1$ in the modified envy graph, we use 
\PR subroutine and allocate $Y$ to $i$ to update the allocation, when agent $i$ prefers $Y$ of $k+1$ goods from $X_s$ and the pool to her own bundle.
}}
\setcounter{AlgoLine}{8}
\nl \Kw{\texttt{Step 9}}\label{step9*}
\Else {\Break\;}
\Return $\X$ 
\end{algorithm}

\noindent Before moving on, we present some bounds for the the value of the goods in $\mathcal{P}(\pX^1)$. The proofs of these bounds are the same as those in \cite{amanatidis2024pushing}. \medskip

\noindent For any agent $i$ with $|X_i^1|=1$ and $Y\subseteq \PP(\X^1) \ \text{with} \ |Y|=2$, due to Step~\ref{step7}, we get
\begin{equation}\label{eq:X1S2}
    v_i(Y)\leq\frac{2}{3}v_i(X^1_i)
\end{equation}

\noindent For any agent $i$ who has a critical good $g$ in $\mathcal{P}(\pX^1)$, for any other good $g'$ in $\PP(\pX^1)\setminus\{g\}$ it holds that $v_i(g')<\frac{2}{3}v_i(X^1_i)-\frac{1}{2}v_i(X^1_i)=\frac{1}{6}v_i(X^1_i)$. 
Hence for any $Y\subseteq \PP(\X^1) \ \text{with} \ |Y|=3$ we get
\begin{equation}\label{eq:X1S3}
    v_i(Y)< \frac{2}{3}v_i(X^1_i)+\frac{1}{6}v_i(X^1_i)= \frac{5}{6}v_i(X^1_i)
\end{equation}

\noindent For any agent $i$ with $|X_i^1|=2$, if there is a path in the modified envy graph $\Gr(\X^1)$ from $s$ to $i$, due to Step~\ref{step8*}, we get for any $Y\subseteq X^1_s\cup \PP(\X^1)$ with $|Y| = 2$ that
\begin{equation}\label{eq:X2S2}
    v_i(Y)\leq v_i(X^1_i)
\end{equation}

\subsection{Allocating the Contested Critical Goods}\label{sec:acc}

\noindent In this subsection, we present the allocation of the contested critical goods in $\PP(\X^1)$ by Algorithm~\ref{alg:acc}. Let $N_s$ denote the  source agents in the modified envy graph $\Gr(\X^1)$, with $n_s = |N_s|$, and let $M_c$ denote the set of contested critical goods in $\PP(\X^1)$, with $m_c = |M_c|$. \medskip

\noindent By \cref{lem:GPA+-satisfies-properties}, there is at least one source agent in $\Gr(\X^1)$, i.e. $n_s\geq 1$ and this source agent does not have any critical goods. Note that each contested critical good is critical for at least two agents, and each agent can have at most one critical good.
Thus, when there are at most eight agents and at least one agent has no critical good, the number of contested critical goods is at most three, i.e. $m_c \leq 3$.
Furthermore, we only employ Algorithm~\ref{alg:acc}, when there is at least one contested critical good, i.e. $m_c\geq 1$. \medskip

\noindent Given the values of $m_c$ and $n_s$, all instances with $m_c\geq 1$ and $n\leq 8$ are distinguished into the following $5$ cases:
\begin{itemize}[labelwidth=1.2cm, leftmargin=!]
    \item[\textit{(Case 1)}] The number of source agents is at least the number of contested critical goods.
    \item[\textit{(Case 2)}] There is one source agent and two contested critical goods. 
    \item[\textit{(Case 3)}] There are two source agents and three contested critical goods.
    \item[\textit{(Case 4)}] There is one source agent, three contested critical goods, and every agent other than the source agent in $\Gr(\X)$ is allocated exactly one good.
    \item[\textit{(Case 5)}] There is one source agent, three contested critical goods and an agent $j$ who is not the source agent  in $\Gr(\X)$ but is allocated two goods.
\end{itemize}

\noindent In Algorithm~\ref{alg:acc}, we use the simple strategies to allocate contested  critical goods in Cases 1–4. Case 5 is more involved. Hence, we employ a more detailed allocation process in Algorithm~\ref{alg:smallcase5}.
\medskip

\noindent Before introducing Algorithm~\ref{alg:smallcase5}, let $s$ denote the only source agent in the modified envy graph.
Since $m_c = 3$ and $n_s = 1$, there is exactly one agent other than $s$ who has two goods in $\X^1$; denote this agent as $j$.
In Algorithm~\ref{alg:smallcase5}, we divide Case 5 into 5 sub-Cases, and each case is triggered if the conditions of all the previous ones are not satisfied. For the tests in Case~5.1, fix a deterministic rule for choosing
cycles in \ACR and apply it to a separate copy of the input
allocation for each of the three pairs $Y\subseteq M_c$.
Use the same rule in the actual calls to \ACR in Cases~5.1
and~5.2. These trials take polynomial time and do not alter
the input allocation.
\begin{itemize}[labelwidth=1.5cm, leftmargin=!]
    \item[\textit{(Case 5.1)}] There exists a pair of goods $Y \subseteq M_c$ such that, after allocating $Y$ to $s$ and resolving the cycles in the envy graph, a source agent occurs who has either 1 or 4 goods.
    \item[\textit{(Case 5.2)}] There exists an agent who considers a good $g$ in $M_c$ critical and is $2/3$-EFX towards $X_j \cup {g}$. 
    \item[\textit{(Case 5.3)}] There exists a set of two goods $Y\subseteq M_c$, such that there exists an agent $k$ with $|X_k|=1$ who is not envied by any agent in $N\setminus\{s,j\}$ and $v_s(X_k)\leq \frac{3}{2}v_s(X_s\cup Y)$.
    \item[\textit{(Case 5.4)}] There exists set of two goods $Y\subseteq M_c$ such that for agent $j$, $v_j(X_s\cup Y)\leq v_j(X_j)$.
    \item[\textit{(Case 5.5)}] Previous cases are not satisfied.
\end{itemize}

\begin{algorithm}[!htbp]
\DontPrintSemicolon
\caption{\ACC($(v_i)_{i \in N}, \pX$)}
\label{alg:acc}
\SetKwComment{Comment}{/* }{ */}
\SetKw{Continue}{continue}
\SetKw{Break}{break}
\SetKw{Step}{Step}
\SetKw{Substep}{Substep}
\SetKwData{Kw}{}
\KwIn{The values $v_i(g)$ for every agent $i \in N$ and every good $g \in M$, and a partial allocation $\pX^1$ of size at most 2 which satisfies \cref{prop:a*,prop:b*,prop:c*,prop:d*,prop:e*,prop:f*}.}
\KwOut{A partial allocation $\pX^2$ with no contested critical good in $\mathcal{P}(\pX^2)$.}
$N_s \gets \{i \in N:i \text{ is a source agent of the modified envy graph } \Gr(\X^1)\}$\\
$n_s \gets |N_s|$ \\
$M_c \gets \{g \in \PP(\X^1): g \text{ is a contested critical good}\}$\\
$m_c \gets |M_c|$ \\
\If{ $n_s\geq m_c$}{
Let $N' \subseteq N_s$ such that $|N'|=m_c$\\
\For{$s_t \in N'$}{
Choose a good $g_t$ in $M_c$ arbitrarily.\\
$X_{s_t}\gets X_{s_t}\cup\{g_t\}$\\
$M_c\gets M_c\setminus \{g_t\}$\\
}
\tcp{Case 1: $n_s\geq m_c$. We allocate each contested critical good to a source agent in $\Gr(\X)$ individually.}
}
\ElseIf{$n_s=1$ and $m_c=2$
}{
 Let $s$ a source agent in $\Gr(\pX)$ \\
 $X_s\gets X_s\cup M_c$\\
 \tcp{Case 2: $n_s=1$ and $m_c=2$. We allocate the two contested critical goods to the source agent in $\Gr(\X)$}
}
\ElseIf{
$n_s=2$ and $m_c=3$
}{
Let $s_1$ and $s_2$ two source agents in the modified envy graph $\Gr(\X)$ \\
Choose two goods $g_1,g_2$ in $M_c$ arbitrarily.\\
$X_{s_1}\gets X_{s_1}\cup \{g_1\}$\\
$X_{s_2}\gets X_{s_2}\cup \{g_2\}$\\
$M_c\gets M_c\setminus\{g_1,g_2\}$\\
\If{There are some cycles in the envy graph $G(\pX)$}{
$\X \gets\ACR(\X,G)$\;
}
Let $s'$ source agent in the envy graph $G(\X)$.\\
$X_{s'}\gets X_{s'}\cup M_c$\\
\tcp{Case 3: $n_s = 2$ and $m_c = 3$. Allocate one contested critical good from $M_c$ to each source agent in the modified envy graph and resolve any cycles if they exist in the envy graph. Then, allocate the remaining goods in $M_c$ to one of the source agents in the envy graph.}
}
\ElseIf{
$n_s=1$, $m_c=3$ and $|X_i|=1$ for any agent $i$ who is not source agent in $\Gr(\X)$
}{
Choose two goods $\{g_1,g_2\}$ in $M_c$ arbitrarily. \\
Let $s$ a source agent in the modified envy graph $\Gr(\X)$\\
$X_s\gets X_s\cup \{g_1,g_2\}$\\
$M_c\gets M_c\setminus \{g_1,g_2\}$\\
\If{There are some cycles in the envy graph $G(\pX)$}{
$\X \gets\ACR(\X,G)$\;
}
Let $s'$ a source agent in the envy graph $G(\X)$ \\
$X_{s'}\gets X_{s'}\cup M_c$\\
\tcp{Case 4: $n_s=1$, $m_c=3$ and only the source agent in the modified envy graph has two goods. Allocate two contested critical goods from $M_c$ to the source agent in the modified envy graph and resolve any cycles if they exist in the envy graph. Then, allocate the remaining good in $M_c$ to one of the source agents in the envy graph.}
}
\Else{
$\X\gets \textsc{LastAllocateContestedCritical}((v_i)_{i \in N}, \pX)$\;
\tcp{Case 5: None of the conditions in Cases 1–4 are satisfied. This case is more involved, so we defer to Algorithm~\ref{alg:smallcase5} for allocating the contested critical goods.}
}
\Return $\X$ 
\end{algorithm}

\noindent In every case, we allocate all the contested critical goods in $\PP(\X^1)$.
Now we prove that the partial allocation $\X^2$ is $2/3$-EFX.

\begin{lemma}\label{lem:acc}
   Let $\X^0$ be a seed allocation, and let $\X^1 = \GPPAplus((v_i){i \in N},\X^0,1)$.
    Then the allocation $\X^2$, which is the output of $\ACC((v_i){i \in N}, \X^1)$, is $2/3$-EFX.
\end{lemma}
\begin{proof}
\noindent \emph{If $\X^2$ is computed in Case 1: } The allocation $\X^2$ guarantees that $v_i(X^2_i)\geq v_i(X^1_i)$ for every $i\in [n]$. Suppose that agent $s_t$ receives a new good $g_t$ in the algorithm. For any agent $i$ who does not have a critical good in $\PP(\X^1)$ holds that
\[
v_i(X^1_{s_t}\cup \{g_t\})\leq v_i(X^1_i)+\frac{1}{2}v_i(X^1_i)\leq \frac{3}{2}v_i(X^2_i).
\]
And for an agent $i$ who has a critical good in $\PP(\pX^1)$, 
\[
v_i(X^1_{s_t}\cup \{g_t\})\leq \frac{2}{3}v_i(X^1_i)+\frac{5}{6}v_i(X^1_i)\leq\frac{3}{2}v_i(X^2_i).
\]
Thus, $\X^2$ computed in Case 1 is $\frac{2}{3}$-EFX.\medskip

\noindent \emph{If $\X^2$ is computed in Case 2: }
The allocation $\X^2$ guarantees that $v_i(X^2_i)\geq v_i(X^1_i)$ for every $i\in [n]$. Assume that agent $s$ receives goods $g_1$ and $g_2$ in $M_c$. For agent $i$ with $|X^1_i|=1$, by inequality (\ref{eq:X1S2}), we have 
\[
 v_i(X^1_s\cup \{g_1,g_2\})\leq \frac{2}{3}v_i(X^1_i)+\frac{2}{3}v_i(X^1_i)\leq \frac{3}{2}v_i(X^2_i).
\]
For agent $i$ who does not have a critical good in $\PP(\X^1)$ and $|X^1_i|=2$; since agent $s$ is the only source agent in the modified envy graph $\Gr(\X^1)$, there must exist a path from agent $s$ to agent $i$.
By inequality (\ref{eq:X2S2}), agent $i$ values any two goods in $X^1_s\cup\{g_1,g_2\}$ at most $v_i(X^1_i)$. 
Hence, for agent $i$, the value of any three goods in $X^1_s\cup \{g_1,g_2\}$ is no larger than $\frac{3}{2}v_i(X^1_i)$ and is also no larger than $\frac{3}{2}v_i(X^2_i)$.
Thus, $\X^2$ computed in Case 2 is $\frac{2}{3}$-EFX.\medskip

\noindent \emph{If $\X^2$ is computed in Case 3: }
The allocation $\X^2$ guarantees that $v_i(X^2_i)\geq v_i(X^1_i)$ for every $i\in [n]$.
Good $g_1$ is allocated to agent $s_1$, good $g_2$ is allocated to agent $s_2$.
Let $g_3$ be the last good in $M_c\setminus\{g_1,g_2\}$.
Since $n_s=2$, every agent who is not source agent in $\Gr(\X^1)$ contains only one good.
We divide this case into two sub-cases depending on which bundle, $g_3$ is allocated. \medskip

(\emph{Case 3.1}) Good $g_3$ is allocated to the bundle $X^1_{s_1}\cup\{g_1\}$ or $X^1_{s_2}\cup \{g_2\}$.
Without loss of generality, suppose that $g_3$ is allocated to the bundle $X^1_{s_1}\cup \{g_1\}$.
For any agent $i$ who does not have a critical good in $\PP(\X^1)$, since $g_3$ is allocated to the source agent in the envy graph, we get
\[v_i(X^1_{s_1}\cup\{g_1,g_3\})\leq v_i(X^2_i)+v_i(g_3)\leq  v_i(X^2_i)+\frac{1}{2}v_i(X^1_i)\leq \frac{3}{2}v_i(X^2_i).
\]
For any agent $i$ who has one critical good in $\PP(\pX^1)$, by inequality (\ref{eq:X1S2}), we get
\[v_i(X^1_{s_1}\cup\{g_1,g_3\})\leq \frac{2}{3}v_i(X^1_{i})+v_i(\{g_1,g_3\})\leq \frac{3}{2}v_i(X^1_i)\leq \frac{3}{2}v_i(X^2_i).
\]

(\emph{Case 3.2}) Good $g_3$ is allocated to one bundle containing only one good, say $g_k$. For agent $i\in N$, as proved in Case 1, holds that
$v_i(X^1_{s_1}\cup \{g_1\})\leq \frac{3}{2}v_i(X^2_i)$ and $ v_i(X^1_{s_2}\cup \{g_2\})\leq \frac{3}{2}v_i(X^2_i)$
Since $g_k$ is in a bundle not envied by any other agent, we have $v_i(g_k)\leq v_i(X^2_i)$. By \cref{prop:d*}, it holds that $v_i(g_3)\leq v_i(X^1_i)\leq v_i(X^2_i)$. Therefore, every agent $i\in N$ is EFX towards $\{g_3,g_k\}$ and
 $\X^2$ computed in Case 3 is $\frac{2}{3}$-EFX.\medskip

\noindent \emph{If $\X^2$ is computed in Case 4: }
The allocation $\X^2$ guarantees that $v_i(X^2_i)\geq v_i(X^1_i)$ for every $i\in [n]$. Goods $g_1$ and $g_2$ are allocated to agent $s$. Let $g_3$ be the last good in $M_c\setminus\{g_1,g_2\}$.
We divide this case into two sub-cases depending on which bundle $g_3$ is allocated: \medskip

(\emph{Case 4.1}) Good $g_3$ is allocated to $X^1_s\cup\{g_1,g_2\}$.
For any agent $i$ with one critical good in $\PP(\X^1)$, by inequality (\ref{eq:X1S3}) we get
\[v_i(X^1_{s}\cup \{g_1,g_2,g_3\})\leq \frac{2}{3}v_i(X^1_i)+\frac{5}{6}v_i(X^1_i)\leq \frac{3}{2}v_i(X^2_i).\]
For any agent $i$ with no critical good in $\PP(\X^1)$,  since $g_3$ is allocated to the source agent in the envy graph, we get
\[v_i(X^1_{s}\cup \{g_1,g_2,g_3\})\leq v_i(X^1_{s}\cup \{g_1,g_2\})+v_i(g_3)\leq v_i(X^2_i)+\frac{1}{2}v_i(X^1_i)\leq \frac{3}{2}v_i(X^2_i).\]
\medskip

(\emph{Case 4.2}) Good $g_3$ is allocated to one bundle with only one good, say $g_k$.
As proved in Case 2, for any agent $i\in N$, $i$ is $\frac{2}{3}$-EFX towards $X^1_{s}\cup\{g_1,g_2\}$.
As proved in Case 3.2, for any agent $i\in N$, $i$ is EFX towards $\{g_3,g_k\}$.
Hence, $\X^2$ computed in Case 4 is $\frac{2}{3}$-EFX.\medskip

\noindent Before we proceed to the analysis, we have the following observation for Case 5. There are 3 contested critical goods in $\PP(\X^1)$, an agent $j$  who is not the source agent but has two goods and no more than $8$ agents in the instance, hence, every good in $M_c$ is critical for two agents.
\medskip

\noindent \emph{If $\X^2$ is computed in Case 5.1:} 
The allocation $\X^2$ guarantees that $v_i(X^2_i)\geq v_i(X^1_i)$ for every $i\in [n]$.
As proved in Case 4, for any agent $i\in N$, $i$ is $2/3$-EFX towards any bundle in $\X^2$.\medskip

\noindent \emph{If $\X^2$ is computed in Case 5.2:} 
The allocation $\X^2$ guarantees that $v_i(X^2_i)\geq v_i(X^1_i)$ for every $i \in [n]$. As proved in Case 2, every agent is $2/3$-EFX towards $X_s\cup Y$.
Since in this case, Case 5.1 is not satisfied, good $g$ is allocated to bundle $X^1_j$.  For any agent $i$ who does not have good $g$ as critical, since $g$ is allocated to the source agent in the envy graph, we get 
\[
v_i(X^1_j\cup \{g\})\leq v_i(X^2_i)+\frac{1}{2}v_i(X^1_i)\leq \frac{3}{2}v_i(X^2_i).
\]
Since $g$ is critical for two agents and one agent is $2/3$-EFX towards $X_j\cup \{g\}$, there is at most one agent who is not $2/3$-EFX towards  $X_j\cup \{g\}$.
Assume that this agent is $k$, if such an agent exists.
Since there is only one source agent $s'$ in the envy graph before allocating $g$, there must exist a path from $s'$ to $k$. Thus, $\X^2$ computed in Case 5.2 is $\frac{2}{3}$-EFX.\medskip

\noindent For the following three sub-cases, the conditions of Case  5.1 and 5.2 are not satisfied. Hence, for any agent $i$ who has a critical good $g$ in $\PP(\X^1)$, we have that $|X^1_i|=1$ and 
\[v_i(X^1_{j})\geq \frac{3}{2}v_i(X^1_i)-v_i(g)\geq \bigg(\frac{3}{2}-\frac{2}{3}\bigg)v_i(X^1_i)>\frac{2}{3}v_i(X^1_i).\]
Hence, for agent $j$ it holds that, $v_j(X^1_i)\leq \frac{3}{2}v_j(X^1_j)$.
Otherwise, there exists a cycle between $j$ and $i$ in the modified envy graph $\Gr(\X^1)$ and Step~\ref{step7} in Algorithm~\ref{alg:G3PA+} should have been executed, which is a contradiction.\medskip

\noindent \emph{If $\X^2$ is computed in Case 5.3:}  
The allocation $\X^2$ guarantees that $v_i(X^2_i)\geq v_i(X^1_i)$ for every $i\in [n]$. Let $g_k$ be the last good in $M_c$.
As shown in Case 2, every agent is $2/3$-EFX towards $X^1_s\cup Y$.
For agent $j$, since $g$ is not critical for her and $v_j(X^1_k)\leq \frac{3}{2}v_j(X^1_j)$ holds, $j$ is $2/3$-EFX towards $X^1_k\cup \{g\}$ in $\X^2$.
For agent $s$, the condition of Case~5.3 gives
$v_s(X_k^1)\leq\frac32v_s(X_s^1\cup Y)
=\frac32v_s(X_s^2)$, while
$v_s(g)\leq\frac12v_s(X_s^1)\leq\frac32v_s(X_s^2)$
by \cref{prop:e*}. Thus, $s$ is $2/3$-EFX towards
$X_k^1\cup\{g\}$.
For any agent $i\in N\setminus\{s,j\}$, it holds that $v_i(g)\leq v_i(X^1_i)$ and $v_i(X^1_k)\leq v_i(X^1_i)$.
Hence, $i$ is $2/3$-EFX towards $X^1_k\cup \{g\}$ in $\X^2$. Thus
$\X^2$ computed in Case 5.3 is $\frac{2}{3}$-EFX.\medskip

\noindent \emph{If $\X^2$ is computed in Case 5.4:} 
The allocation $\X^2$ guarantees that $v_i(X^2_i)\geq v_i(X^1_i)$ for every $i\in [n]$.
For agent $j$, since she has no critical good, it holds that
$
v_j(X^1_s\cup M_c)=v_j(X^1_s\cup Y)+v_j(M_c\setminus Y)\leq \frac{3}{2}v_j(X^2_j).
$
For any agent $i$ who has a critical good, by inequality (\ref{eq:X1S3}), we get
$
v_i(X^1_s\cup M_c)=v_i(X^1_s)+v_i(M_c)\leq \frac{3}{2}v_i(X^2_i).
$
Therefore, $\X^2$ computed in Case 5.4 is $\frac{2}{3}$-EFX.\medskip

\noindent \emph{If $\X^2$ is computed in Case 5.5:}
First, we show that in $\X^1$, there exists an agent $k$ in $N \setminus \{s, j\}$ who is not envied by another one in the same set. Suppose, for contradiction, that no such agent exists in $\X^1$. Then, for every agent $i$ with $|X^1_i| = 1$, there must be some other agent $t$ with $|X^1_t| = 1$ who envies $i$. This implies the existence of a cycle among agents with singleton bundles in the modified envy graph $\Ge(\X^1)$, which contradicts the fact that Step~\ref{step5} has already resolved all cycles. \medskip

\noindent Since Case~5.5 does not satisfy the conditions of Cases~5.1 or~5.2, agent $k$ is not $2/3$-EFX towards $X^1_j \cup \{g\}$.
Furthermore, since $v_k(g_1) \geq v_k(g_2)$ and $g$ is critical, we have that $v_k(\{g,g_1\}) \geq \frac12 v_k(X^1_j \cup \{g\}) + \frac12 v_k(g) > v_k(X^1_k)$.
Also, Case~5.5 does not satisfy the condition of Case~5.3, and $k$ is not envied by any other agent in $N \setminus \{s, j\}$, hence it follows that $v_s(X^1_k \cup \{g_2\}) \geq v_s(X^1_k) > \frac32 v_s(X^1_s)$.
Additionally, since Case~5.5 does not satisfy the condition of Case~5.4, we have that $v_j(X^1_s \cup Y) > v_j(X^1_j)$.
Therefore, for any $i \in N$, it holds that $v_i(X^2_i) \geq v_i(X^1_i)$. \medskip

\noindent As proved in Case 2, every agent in $N\setminus\{s,j\}$ is $2/3$-EFX towards $X^1_s\cup Y$ in $\X^2$.
Agent $j$ receives this bundle, while agent $s$ does not envy it, since failure of Case~5.3 gives $v_s(X^1_s\cup Y)<\frac23 v_s(X^1_k)\leq\frac23 v_s(X^2_s)$.
As shown above, in Case 5.5, it holds that $v_j(X^1_k)\leq \frac32 v_j(X^1_j)$.
Since $X^1_k$ is not envied by any agent $i\in N\setminus\{s,j\}$, for agent $i$, it holds that $v_i(X^1_k)\leq \frac32 v_i(X^1_i)$.
For any agent $i\in N\setminus\{s,j\}$, by \cref{prop:b*}, we get $v_i(g_1)\leq v_i(X^1_i)$ and $v_i(g_2)\leq v_i(X^1_i)$. \medskip

\noindent For any agent $i\in\{s,j\}$, by \cref{prop:c*} for $i=s$ and by non-negativity for $i=j$, we have $v_i(g_1)\leq \frac32 v_i(X^1_i)$ and $v_i(g_2)\leq \frac32 v_i(X^1_i)$.
For any agent $i\in N$, by \cref{prop:d*}, we have $v_i(g)\leq v_i(X^1_i)$.
Hence, every agent in $N\setminus\{s\}$ is $\frac23$-EFX towards $X^1_k\cup\{g_2\}$ and every agent in $N$ is $\frac23$-EFX towards $\{g_1,g\}$.
Since $X^1_k\cup\{g_2\}$ is allocated to agent $s$, $s$ is EFX towards $X^1_k\cup\{g_2\}$.
Thus, $\X^2$ computed in Case 5.5 is $2/3$-EFX.
\end{proof}

\begin{algorithm}[!htbp]
\DontPrintSemicolon
\caption{\textsc{LastAllocateContestedCritical}($(v_i)_{i \in N}, \pX$)}
\label{alg:smallcase5}
\SetKwComment{Comment}{/* }{ */}
\SetKw{Continue}{continue}
\SetKw{Break}{break}
\SetKw{Step}{Step}
\SetKw{Substep}{Substep}
\SetKwData{Kw}{}
\KwIn{The values $v_i(g)$ for every agent $i \in N$ and every good $g \in M$, and a partial allocation $\pX^1$ of size at most 2 which satisfies \cref{prop:a*,prop:b*,prop:c*,prop:d*,prop:e*,prop:f*}.}
\KwOut{A partial allocation $\pX^2$ with no contested critical good in $\mathcal{P}(\pX^2)$.}

$M_c \gets \{g \in \PP(\X^1): g \text{ is a contested critical good}\}$\\

\If{there exists a subset $Y \subseteq M_c$ of two goods such that, after allocating $Y$ to agent $s$ and resolving the cycles in the envy graph $G(\X^1_{-s}, X_s \cup Y)$, there is a source agent who is allocated either one good or four goods}{
$X_s\gets X_s\cup Y$\\
$M_c\gets M_c\setminus Y$\
\If{there exists cycles in the envy-graph $G(\X)$}{
$\X \gets \text{AllCyclesResolution}(\X, G)$\;
}
Let $s'$ source agent in the envy-graph $G(\X)$ who has either 1 or 4 goods.\\
$X_{s'}\gets X_{s'}\cup M_c$\\
\tcp{Case 5.1: We allocate $Y$ to agent $s$, resolve the cycles in the envy graph, and then allocate the remaining good in $M_c$ to a source agent $s'$ who has 1 or 4 goods.
}
}
\ElseIf{there exists an agent who has a critical good $g\in M_c$
and is $2/3$-EFX towards $X_j^1\cup\{g\}$}{
    $S\gets M_c\setminus\{g\}$\;
    $X_s\gets X_s\cup S$\;
    $\X\gets\ACR(\X,G(\X))$\;
    Let $q$ be the unique source of $G(\X)$\;
    $B\gets X_q$\;

    \If{there exists an agent $k\ne q$ who is not $2/3$-EFX
    towards $B\cup\{g\}$ in the current allocation}{
        Let $\Pi=(i_1,\ldots,i_\ell)$ be a path from $i_1=q$
        to $i_\ell=k$ in $G(\X)$\;
        $(X_{i_h})_{h=1}^{\ell-1}
        \gets\PR(\X,G(\X),\Pi)$\;
        $X_k\gets B\cup\{g\}$\;
    }
    \Else{
        $X_q\gets B\cup\{g\}$\;
    }
    \tcp{Case 5.2: We allocate the two goods in $M_c\setminus\{g\}$ to the source agent $s$, resolve the cycles in the envy graph, and then allocate the remaining good $g$ to the new source agent.
If there is still an agent $k$ who is not $2/3$-EFX towards $X_j \cup {g}$, we allocate $X_j\cup \{g\}$ to agent $k$ and resolve the path from $s$ to $k$ in the envy graph.}
}
\ElseIf{there exists a set of two goods $Y\subseteq  M_c$ and an agent $k$ with $|X_k|=1$ such that no agent in $N\setminus\{s,j\}$ envies $X_k$ and $v_s(X_k)\leq 3/2v_s(X_s\cup Y)$}{
$X_s\gets X_s\cup Y$
$X_k\gets X_k\cup (M_c\setminus Y)$.\\
\tcp{Case 5.3: We allocate $Y$ to agent $s$ and the remaining good in $M_c\setminus Y$ to agent $k$.}
}
\ElseIf{there exists set of two goods $Y\subseteq M_c$ such that for agent $j$, $v_j(X_s\cup Y)\leq v_j(X_j)$}{
 $X_s\leftarrow X_s\cup M_c$\;
 \tcp{Case 5.4: We allocate all the three goods in $M_c$ to agent $s$.}
}
\Else{
Let $k\in N\setminus\{s,j\}$ be an agent not envied by
any other agent in $N\setminus\{s,j\}$.\\
Let good $g$ in $M_c$ that is critical for agent $k$. \\
$Y\gets M_c\setminus\{g\}$\\
Let $g_1$ and $g_2$ the two goods in $X_j$, where agent $k$ prefers $g_1$ than $g_2$.\\
$X_j\gets X_s\cup Y$\\
$X_s\gets X_{k}\cup \{g_2\}$\\
$X_{k}\gets \{g_1,g\}$\\

\tcp{Case 5.5: We first find the agent $k$ who is not envied by any other agent in $N \setminus \{s, j\}$, and identify the good $g \in M_c$ that is critical for $k$, (the remaining goods denoted by $Y = M_c \setminus {g}$). Let $X_j$ be the bundle that contains $g_1$ and $g_2$, where $k$ prefers $g_1$ over $g_2$. Then, we allocate $Y$ together with $X_s$ to agent $j$, allocate $g_2$ together with $X_{k_1}$ to agent $s$, and allocate ${g_1, g}$ to agent $k$.
}
}
\Return $\X$ 
\end{algorithm}

\subsection{Allocating the Uncontested Critical Goods and the Remaining Goods}\label{sec:ucc}

The partial allocation $\X^2$ satisfies: (1) $\X^2$ is $2/3$-EFX, (2) each agent has at most one critical good, and (3) no good is critical for more than one agent. To complete the allocation, we first use Algorithm~\ref{alg:ucc} to allocate the remaining critical goods and get $\X^3$.
Finally, we employ the Envy Cycle Elimination algorithm~\ref{alg:ece} to allocate the remaining goods and get the final complete allocation $\X^4$. $\X^4$ remains a $2/3$-EFX allocation, known from \cite{amanatidis2024pushing}.
Hence, Theorem~\ref{thm:8-agents-23-efx} is proved. \medskip

\SetAlgorithmName{ALGORITHM}{Algorithm}{List of Algorithms}

\begin{algorithm}[h!]
\DontPrintSemicolon
\caption{\textsc{UncontestedCritical}($\pX,\tilde{G}$)}
\label{alg:ucc}
\SetKwComment{Comment}{/* }{ */}
\KwData{A partial allocation $\X$ with no contested critical good in $\PP(\X)$, and its envy graph $\tilde{G}(X)$.\vspace{3pt}}
\KwResult{A partial allocation $\X'$ which does not induce any critical goods.\vspace{3pt}}
$\X \gets \text{AllCyclesResolution}(\X,G(\X))$\;
\While{there exists $i \in N$ and $g_i \in \PP(\X)$ such that $v_i(g_i) > \frac{1}{2}v_i(X_i)$}{
\tcp{while there exists an agent with a critical good}
    Let $s$ be a source of $G(\X)$ with a path
    $\Pi=(i_1,\ldots,i_\ell)$ from $i_1=s$ to $i_\ell=i$\;
    $S\gets X_s$\;
    
    \If{$v_i(S\cup\{g_i\})>v_i(X_i)$}{
        $(X_{i_h})_{h=1}^{\ell-1}\gets\PR(\X,G(\X),\Pi)$\;
        $X_i\gets S\cup\{g_i\}$\;
    }
    \Else{
        $X_s\gets S\cup\{g_i\}$\;
    }
    
    $\PP(\X)\gets\PP(\X)\setminus\{g_i\}$\;
    $\X\gets\ACR(\X,G(\X))$\;
    
}
\Return $\X$
\end{algorithm}


\end{document}